\documentclass[11pt]{article}
\usepackage[letterpaper,margin=1in]{geometry}
\usepackage[T1]{fontenc}
\usepackage[utf8]{inputenc}
\usepackage{lmodern,microtype}
\usepackage{amsmath,amssymb,amsthm,mathtools}
\usepackage{booktabs,array,enumitem,aliascnt,xcolor}
\usepackage[colorlinks=true,linkcolor=blue!50!black,citecolor=blue!50!black,urlcolor=blue!50!black]{hyperref}
\usepackage[nameinlink,capitalise,noabbrev]{cleveref}
\hypersetup{pdftitle={Capacity List Decoding and Mutual Correlated Agreement for Reed--Solomon Codes},pdfauthor={Fernando Granha Jeronimo},pdfsubject={Hidden-derivative interpolation, field-size-independent list bounds, and exact-support proximity gaps}}
\allowdisplaybreaks
\setlist{itemsep=2pt,topsep=5pt}
\newtheorem{theorem}{Theorem}[section]
\newaliascnt{lemma}{theorem}
\newtheorem{lemma}[lemma]{Lemma}
\aliascntresetthe{lemma}
\newaliascnt{proposition}{theorem}
\newtheorem{proposition}[proposition]{Proposition}
\aliascntresetthe{proposition}
\newaliascnt{corollary}{theorem}
\newtheorem{corollary}[corollary]{Corollary}
\aliascntresetthe{corollary}
\theoremstyle{definition}
\newaliascnt{definition}{theorem}
\newtheorem{definition}[definition]{Definition}
\aliascntresetthe{definition}
\theoremstyle{remark}
\newaliascnt{remark}{theorem}
\newtheorem{remark}[remark]{Remark}
\aliascntresetthe{remark}
\newcommand{\F}{\mathbb F}
\newcommand{\A}{\mathbb A}

\newcommand{\Z}{\mathbb Z}
\newcommand{\cA}{\mathcal A}
\newcommand{\cC}{\mathcal C}
\newcommand{\cK}{\mathcal K}
\newcommand{\cL}{\mathcal L}
\newcommand{\cQ}{\mathcal Q}
\newcommand{\cT}{\mathcal T}
\newcommand{\cV}{\mathcal V}
\newcommand{\RS}{\operatorname{RS}}
\newcommand{\Sol}{\operatorname{Sol}}
\newcommand{\Agr}{\operatorname{Agr}}
\newcommand{\Bad}{\operatorname{Bad}}
\newcommand{\errMCA}{\operatorname{err}_{\mathrm{MCA}}}
\newcommand{\cdeg}{\operatorname{cdeg}}
\newcommand{\dist}{\Delta}
\newcommand{\eps}{\varepsilon}
\newcommand{\Span}{\operatorname{span}}
\newcommand{\rank}{\operatorname{rank}}
\newcommand{\abs}[1]{\left|#1\right|}
\newcommand{\ceil}[1]{\left\lceil #1\right\rceil}
\newcommand{\floor}[1]{\left\lfloor #1\right\rfloor}

\title{Algorithmic List Decoding at Capacity  and Optimal Proximity Gaps\\for Reed--Solomon Codes}
\author{Fernando Granha Jeronimo\thanks{{\tt University of Illinois, Urbana-Champaign}. {\tt granha@illinois.edu}.}}
\date{}
\begin{document}
\maketitle
\begin{abstract}
We give a unified hidden-derivative framework for list decoding and mutual correlated agreement of ordinary Reed--Solomon codes over prime fields, on arbitrary prescribed evaluation sets. For every fixed slack $\gamma>0$, every sufficiently large block length $n$, every prime $q\ge n$, and every dimension $1\le k\le(1-\gamma)n$, a deterministic algorithm finds all codewords within relative distance $1-k/n-\gamma$ in $q^{O_\gamma(1)}$ time. The final list has size $n^{O_\gamma(1)}$, independently of $q$. Both statements extend to bounded-input-list recovery, with constants depending additionally on the input-list bound.

For every fixed curve degree $\ell$, at most $n^{O_{\gamma,\ell}(1)}$ parameters on a curve $f_0+zf_1+\cdots+z^\ell f_\ell$ admit a nearby codeword whose exact agreement support is not a maximal jointly explained support of the coefficient words. For lines this gives MCA error $n^{O_\gamma(1)}/q$, with no proximity loss.

The interpolation stage reparameterizes and optimizes the hidden-derivative construction of Brakensiek, Chen, Putterman, Zhang, and Zheng; differential root enumeration uses Kopparty's algorithm. We then prove that a specialization-safe differential equation has a cover by constant-dimensional varieties of polynomial cumulative degree, outside polynomially many parameter values. Intersecting these varieties with equations from the full agreement support yields both the field-size-independent list bound and exact-support MCA.  
\end{abstract}
\clearpage
\begingroup\small\tableofcontents\endgroup
\clearpage
\section{Introduction}\label{sec:intro}
Let $D=\{\alpha_1,\ldots,\alpha_n\}\subseteq\F_q$ consist of distinct points. The Reed--Solomon code
\[
 \RS_{\F_q}(D,k)
 =\{(P(\alpha_1),\ldots,P(\alpha_n)):P\in\F_q[X],\ \deg P<k\}
\]
has rate $k/n$ and minimum distance $n-k+1$~\cite{RS60,Sin64}. List decoding asks for all codewords near a received word, rather than a unique closest codeword. At rate $R$, the large-alphabet capacity benchmark is the radius $1-R$; the fixed-slack objective is efficient decoding at every radius $1-R-\gamma$ with $\gamma>0$ fixed. For ordinary Reed--Solomon codes, the classical general algorithmic guarantee is the Johnson radius $1-\sqrt R$, from the interpolation algorithms of Sudan and Guruswami--Sudan~\cite{Sud97,GS99}.

A second question concerns how nearby codewords behave when the received word varies algebraically. Given coefficient words $f_0,\ldots,f_\ell\in\F_q^n$, consider
\[
 y(z)=\sum_{j=0}^{\ell}z^jf_j.
\]
Ordinary correlated agreement seeks one large set of coordinates on which all the $f_j$ are simultaneously explained by codewords. Mutual correlated agreement (MCA) asks for more: outside a small exceptional set of parameters, the \emph{entire exact agreement support of every nearby codeword} must already be jointly explained. This property was introduced in WHIR~\cite{WHIR25} and is relevant to the soundness of Reed--Solomon proximity protocols~\cite{BCIKS20,ABF26}.

We treat these two questions through the same differential interpolant. Root enumeration gives a decoder. A bounded-complexity description of the solution locus gives a list bound independent of the field size and, in the parametric setting, the stronger exact-support MCA conclusion.

\subsection{Main results}
For lists $S_1,\ldots,S_n\subseteq\F_q$, write
\[
 \cL(D,k,\mathbf S,A)
 :=\{P\in\F_q[X]:\deg P<k,\
       |\{i:P(\alpha_i)\in S_i\}|\ge A\}.
\]
Singleton lists recover ordinary list decoding.

\begin{theorem}[Capacity decoding and field-size-independent lists]\label{thm:capacity}
Fix $\gamma\in(0,1)$ and an integer $L\ge1$. There are effective constants $C_{\rm alg}=C_{\rm alg}(\gamma,L)$, $C_{\rm list}=C_{\rm list}(\gamma,L)$, and $n_0=n_0(\gamma,L)$ such that the following holds. Let $n\ge n_0$, let $q\ge n$ be prime, let $D\subseteq\F_q$ have $n$ distinct points, and let $1\le k\le(1-\gamma)n$. For arbitrary lists $|S_i|\le L$, all polynomials in
\[
 \cL\bigl(D,k,\mathbf S,k+\ceil{\gamma n}\bigr)
\]
can be found deterministically using at most $q^{C_{\rm alg}}$ field operations. Moreover,
\[
 \bigl|\cL\bigl(D,k,\mathbf S,k+\ceil{\gamma n}\bigr)\bigr|
 \le n^{C_{\rm list}}.
\]
In particular, every such Reed--Solomon code is list decodable at radius $1-k/n-\gamma$, with polynomial list size independent of $q$ and deterministic polynomial time whenever $q=n^{O(1)}$.
\end{theorem}

The constants are uniform in the actual dimension $k$, not merely in a prescribed constant rate. For every fixed $R\in(0,1)$ and $\gamma\in(0,1-R)$, the theorem therefore gives decoding at radius $1-R-\gamma$ for every $k\le Rn$. Linear-size prime fields are allowed. The final list-size estimate is stronger than the $q^{O(r)}$ bound on all differential-equation roots: it counts only roots passing the agreement filter. The geometric argument proving this distinction is in \cref{sec:list-bounds}.

For MCA, put $A_\rho=n-\floor{\rho n}$. A support $S\subseteq[n]$ is \emph{jointly explained} by $\mathbf f=(f_0,\ldots,f_\ell)$ if there are degree-below-$k$ polynomials $P_0,\ldots,P_\ell$ such that $P_j(\alpha_i)=f_{j,i}$ for all $i\in S$ and all $j$. A parameter $z$ is bad if some degree-below-$k$ polynomial $P$ has at least $A_\rho$ agreements with $y(z)$, but its full agreement support is not a maximal jointly explained support. We give the formal definitions in \cref{sec:support-definitions}.

\begin{theorem}[Fixed-slack capacity-radius curve MCA]\label{w:thm:mca}
Fix $\gamma\in(0,1)$ and an integer $\ell\ge1$.  There are effective constants
\[
 C=C(\gamma,\ell)>0,
 \qquad n_0=n_0(\gamma,\ell),
\]
such that the following holds.  Let $n\ge n_0$, let $q\ge n$ be prime, let $D=\{\alpha_1,\ldots,\alpha_n\}\subseteq\F_q$ have $n$ distinct points, and let
\[
 \cC=\RS_{\F_q}(D,k),
 \qquad 1\le k\le(1-\gamma)n.
\]
For arbitrary $f_0,\ldots,f_\ell\in\F_q^n$ and every
\[
 0\le\rho\le1-\frac{k}{n}-\gamma,
\]
one has
\[
 \abs{\Bad_{\cC,\rho}^{(\ell)}(f_0,\ldots,f_\ell)}
 \le n^{C}.
\]
In particular, for lines,
\[
 \errMCA(\cC,\rho)\le \frac{n^{C(\gamma,1)}}{q}.
\]
\end{theorem}

The two main estimates have different field-size implications. The decoder applies already for $q=\Theta(n)$. The MCA estimate becomes nonvacuous only when $q$ exceeds its polynomial exception budget. A constant-degree curve with more than that many nearby points has a jointly explained support of size at least $A_\rho$ and is consequently close to the code at every parameter; see \cref{w:cor:ordinary-pg}.

\subsection{The common interpolation theorem}
The interpolation stage depends on the total incidence
\[
 M=\sum_{i=1}^n|S_i|.
\]
Fix a multiplicative slack $\vartheta\in(0,1)$ and a padded degree ceiling $K$ with $A/(K-1)\ge(1-\vartheta)^{-1}$. We construct a nonzero differential interpolant of order $r$ whenever
\begin{equation}\label{eq:intro-incidence}
 M\le a_\vartheta(K-1)\frac{r^\vartheta}{\log^4(er)}.
\end{equation}
Every degree-below-$K$ polynomial with $A$ list hits satisfies
\[
 Q(X,P,P^{[1]},\ldots,P^{[r]})\equiv0.
\]
Here $P^{[j]}$ denotes the $j$th Hasse derivative. The order $r$ is free and can be taken sufficiently large, but constant, at any fixed capacity slack. Degree padding makes $K=\Theta_\gamma(n)$ even when the actual code dimension is much smaller.

The construction follows the hidden-derivative mechanism of Brakensiek, Chen, Putterman, Zhang, and Zheng~\cite{BCPZZ26}. The two-cutoff derivative-tail space, the backward-Taylor substitution, and the local-kernel mechanism are theirs. Their presentation fixes the derivative order as a function of the target agreement at the start of its interpolation section~\cite[Eq.~(10)]{BCPZZ26}. We establish the corresponding estimates with that order left free, then optimize the two cutoffs. The resulting dimension-to-rank ratio is
\[
 \Omega_\vartheta\!\left((K-1)\frac{r^\vartheta}{\log^4(er)}\right).
\]
We also stack constraints over arbitrary point--symbol incidences and count the existing monomial space by weak compositions. These give sparse list recovery and a $q^{O_\vartheta(r)}$ interpolation algorithm. Prime-field differential root enumeration is due to Kopparty~\cite{Kop15}, in the formulation recorded in~\cite{BCPZZ26}.

\subsection{From differential equations to exact supports}
The geometric part uses more information than a black-box list-size bound. We first run interpolation over $\F_q(Z)$ and choose a primitive polynomial kernel vector. This produces one polynomial $Q(Z,X,Y_0,\ldots,Y_r)$ that remains nonzero after every specialization $Z=z$ and contains every nearby-polynomial witness in its specialized solution locus.

The key algebraic statement, \cref{w:thm:jet-decomposition}, covers these witnesses outside polynomially many parameter values by varieties of dimension at most $r+1$ and polynomial cumulative degree. At an anchor $\alpha$, the Hasse--Taylor equations recursively recover all coefficients beyond the initial jet. Their denominators have exponents $O(d)$ and their numerator degrees are polynomial in the message degree $d$ and the coefficient degree of $Q$. Singular solutions are handled by induction on total derivative-variable degree. Thus the cover controls degree as well as dimension.

For a fixed received word, agreement at $k$ distinct coordinates uniquely determines a polynomial. A candidate cannot lie on a positive-dimensional branch satisfying every equation from its agreement support. Successive support-hyperplane cuts therefore isolate it after at most $r$ dimension drops. B\'ezout and the number of possible index tuples yield the field-size-independent list bound. With input lists, the available cuts are the $M$ point--symbol equations.

For a curve of received words, we instead follow each witness's \emph{exact} support. A bad witness reaches a zero-dimensional or vertical component after constantly many proper cuts. Such components contribute only polynomially many parameter values. On a nonvertical branch where all support equations vanish identically, a Vandermonde system expresses the polynomial coefficients as a degree-$\ell$ polynomial in $z$. Transcendence of the parameter splits every remaining support equation coefficientwise, giving the coefficient codewords over $\F_q$. Reed--Solomon uniqueness then proves maximality of the full support.

The symbolic-specialization and jet-decomposition setup develops the framework of the unpublished ordinary-proximity companion note~\cite{PG26}; all geometric statements used here are proved in this paper, so no result is invoked from that note. The support-preserving descent, and its use for both candidate counting and MCA, make the relationship between the two conclusions explicit.

\subsection{Context and scope}
Related algebraic codes attain capacity through folding or derivative evaluations~\cite{PV05,GR08,GW13,Kop15,KRSW23,GHKS24,CZ25,AHS26}. Earlier work improved list decoding for suitable ordinary Reed--Solomon evaluation sets~\cite{RW14,ST23,GLS21,GST23}. Generic and random evaluation sets admit capacity list-size guarantees, including over linear-size fields~\cite{BGM24,GZ23,AGL24}; \cite{AGG25} is the merged journal version of~\cite{GZ23,AGL24}. Our setting retains the ordinary encoding map and allows every prescribed evaluation set over a prime field. BCPZZ supplies the capacity decoder in the low constant-rate regime~\cite{BCPZZ26}; the all-rate conclusion here extends that framework, rather than replacing its interpolation mechanism.

At the Johnson radius, ordinary proximity gaps are studied in~\cite{BCIKS20,BCCHKS25}, and MCA in~\cite{Hab25,BCGM26}. Fixed integer improvements beyond Johnson are obtained in~\cite{Jo26}. Goyal and Guruswami prove capacity-approaching proximity gaps and MCA for subspace-design codes and random Reed--Solomon codes over linear-size fields~\cite{GG25}. Our evaluation-set guarantee is uniform, at the expense of a much larger field-size budget for small MCA error. Counterexamples near capacity~\cite{CS25,KKH26} require care with shrinking slack: our constants depend on a fixed $\gamma$, and the theorems do not assert a uniform estimate when $\gamma$ tends to zero with $n$.

The result is not a claim of $\operatorname{poly}(n,\log q)$ decoding time for arbitrary large fields, constant output lists, or practical dependence on $1/\gamma$. It does give fixed-slack asymptotic certificates for both the MCA and fixed-interleaving list-size targets in the preliminary Proximity Prize statement~\cite{ABF26,Prize26}; it does not determine their largest safe radius. These consequences and limitations are stated in \cref{sec:consequences}.

\paragraph{Organization.}
\Cref{sec:preliminaries} sets notation. \Cref{sec:interpolation} proves the common interpolation theorem and its specialization-compatible matrix formulation. \Cref{sec:decoding} gives the uniform parameter choice and the sparse decoder. \Cref{sec:geometry} proves the symbolic construction and the jet theorem. \Cref{sec:list-bounds} establishes the polynomial list bound and completes \cref{thm:capacity}. \Cref{sec:mca-proof} proves exact-support MCA. \Cref{sec:consequences} records interleaving, coordinate transformations, and quantitative certificates. The appendices give the root-solver transfer and the more detailed parameter estimates.

\section{Preliminaries}\label{sec:preliminaries}

Throughout, $\log$ denotes the natural logarithm. A zero polynomial satisfies every displayed upper degree bound. We write $[n]=\{1,\ldots,n\}$.  For $u,v\in\F_q^n$, their relative Hamming distance is
\[
  \Delta(u,v):=\frac{1}{n}\abs{\{i\in[n]:u_i\ne v_i\}}.
\]
A code $C\subseteq\F_q^n$ is $(\rho,L)$ list decodable if every Hamming ball of relative radius $\rho$ contains at most $L$ codewords.  In list recovery, the input is a sequence of sets $S_1,\ldots,S_n\subseteq\F_q$, and a codeword is retained when its $i$th symbol lies in $S_i$ on the prescribed number of coordinates.  All polynomial degrees are ordinary univariate degrees unless a weighted degree is displayed explicitly.

\subsection{Hasse derivatives}

For $P\in F[X]$ over a field $F$, the $j$th Hasse derivative $P^{[j]}$ is defined by
\begin{equation}\label{v:eq:hasse}
  P(X+T)=\sum_{j\ge0}P^{[j]}(X)T^j.
\end{equation}
Replacing $X$ by $X+T$ and $T$ by $-T$ gives the backward identity
\begin{equation}\label{v:eq:backward}
  P(X)=\sum_{j\ge0}P^{[j]}(X+T)(-T)^j.
\end{equation}
If $P(\alpha)=y$, then for every $r\ge1$,
\begin{equation}\label{v:eq:truncated-backward}
  P(\alpha+T)
  =y+\sum_{j=1}^{r}(-1)^{j+1}T^jP^{[j]}(\alpha+T)
   +T^{r+1}G_P(T)
\end{equation}
for some $G_P\in F[T]$.

If $P(\alpha+U)=\sum_{j=0}^d c_jU^j$, then $c_j=P^{[j]}(\alpha)$ and
\begin{equation}\label{w:eq:hasse-series}
 P^{[i]}(\alpha+U)=\sum_{j=i}^{d}\binom ji c_jU^{j-i}.
\end{equation}
When $d<\operatorname{char}(F)$, the binomial coefficients with $0\le i\le j\le d$ are nonzero.

\subsection{Sparse list recovery and padding}

\begin{definition}[Sparse list-recovery instance]
A sparse list-recovery instance consists of distinct $\alpha_1,\ldots,\alpha_n\in F$, finite sets $S_1,\ldots,S_n\subseteq F$ over a field $F$, an actual degree bound $k-1$, and an agreement threshold $A$.  Its incidence size is
\[
  M=\sum_{i=1}^n\abs{S_i}.
\]
A polynomial $P$ is valid if $\deg P<k$ and
\[
  \abs{\{i\in[n]:P(\alpha_i)\in S_i\}}\ge A.
\]
\end{definition}

Ordinary list decoding is the case $S_i=\{y_i\}$.  Empty lists and coordinate-dependent list sizes are allowed.

\begin{lemma}[Degree padding]\label{v:lem:padding}
Suppose an algorithm outputs every polynomial of degree below $K$ satisfying a given list-agreement condition.  For every $k\le K$, filtering its output by $\deg P<k$ solves the same problem for degree bound $k-1$.
\end{lemma}

\begin{proof}
Every valid degree-below-$k$ polynomial is among the degree-below-$K$ outputs, and the final filter removes only extraneous roots.
\end{proof}

\subsection{Agreement supports and mutual correlated agreement}\label{sec:support-definitions}

Fix $\cC=\RS_{\F_q}(D,k)$ and a radius $\rho\in[0,1]$ such that $A_\rho:=n-\floor{\rho n}\ge k$. For $u\in\F_q^n$ and $P\in\F_q[X]$, define
\[
 \Agr(u,P):=\{i\in[n]:u_i=P(\alpha_i)\}.
\]
For coefficient words $\mathbf f=(f_0,\ldots,f_\ell)$, call a set $S\subseteq[n]$ \emph{$\mathbf f$-correlated} if there are polynomials $P_0,\ldots,P_\ell\in\F_q[X]$ of degree below $k$ such that
\[
 P_j(\alpha_i)=f_{j,i}
 \qquad (i\in S,\ 0\le j\le\ell).
\]
For these coefficient words write $y(z)=\sum_{j=0}^{\ell}z^jf_j$. Let $\cA_{\rho}(\mathbf f)$ be the collection of inclusion-maximal $\mathbf f$-correlated sets of cardinality at least $A_\rho=n-\floor{\rho n}$.

\begin{definition}[Bad curve parameters]\label{w:def:bad-curve}
For $y(z)=\sum_{j=0}^{\ell}z^jf_j$, define
\[
 \Bad_{\cC,\rho}^{(\ell)}(\mathbf f)
 :=\left\{z\in\F_q:\begin{array}{l}
 \text{there exists }P\in\F_q[X],\ \deg P<k,\text{ with }|\Agr(y(z),P)|\ge A_\rho,\\[-1mm]
 \text{and }\Agr(y(z),P)\notin\cA_\rho(\mathbf f)
 \end{array}\right\}.
\]
For $\ell=1$, the MCA error is
\[
 \errMCA(\cC,\rho)
 :=\max_{f_0,f_1\in\F_q^n}
   \frac{|\Bad_{\cC,\rho}^{(1)}(f_0,f_1)|}{q}.
\]
\end{definition}

In the regime $A_\rho\ge k$, \cref{w:lem:exact-maximal} shows that an exact support of size at least $A_\rho$ is correlated if and only if it belongs to $\cA_\rho(\mathbf f)$. For lines, the definition therefore agrees with the usual MCA condition~\cite{WHIR25,GG25}; for curves, we use the analogous exact-support formulation.

\begin{lemma}[Explained exact supports are maximal]\label{w:lem:exact-maximal}
Let $z\in\F_q$ and let $P\in\F_q[X]$ have degree below $k$.  Put
\[
 S=\Agr(y(z),P).
\]
If $|S|\ge k$ and $S$ is $\mathbf f$-correlated, then $S$ is inclusion-maximal among $\mathbf f$-correlated sets.  In particular, if $|S|\ge A_\rho\ge k$, then $S\in\cA_\rho(\mathbf f)$.
\end{lemma}

\begin{proof}
Choose degree-below-$k$ polynomials $P_0,\ldots,P_\ell$ explaining $S$.  Then $P$ and $\sum_{j=0}^{\ell}z^jP_j$ agree at every $\alpha_i$ with $i\in S$.  Since $|S|\ge k$, Reed--Solomon uniqueness gives
\[
 P=\sum_{j=0}^{\ell}z^jP_j.
\]
Suppose that a strict superset $T\supsetneq S$ were $\mathbf f$-correlated, witnessed by $Q_0,\ldots,Q_\ell$.  The polynomials $P$ and $\sum_jz^jQ_j$ agree on $S$, hence are equal.  The latter polynomial agrees with $y(z)$ on all of $T$, contradicting the definition of $S$ as the \emph{full} agreement support of $P$.  Thus no such $T$ exists.
\end{proof}

\section{Hidden-derivative interpolation}\label{sec:interpolation}
All constructions in this section are over an arbitrary coefficient field $F$. In particular, they apply over $\F_q(Z)$. We prove the parameterized estimates rather than invoking the capacity decoder as a black box.

\subsection{The interpolation space}\label{v:sec:interpolation-space}

We now define the family from which the interpolating polynomial will be chosen.  It is a reparameterization of the two-cutoff interpolation space in~\cite[Eq.~(18)]{BCPZZ26}: the derivative tail obeys both a weighted cutoff, adapted to the backward-Taylor substitution, and an ordinary-degree cutoff.  The gap between the two cutoffs is optimized later in \cref{v:sec:global-interpolation}, and the ordinary cutoff also permits a sharper monomial enumeration.

Fix integers $A,K,r,m$ with $1\le r<K\le A$ and write
\[
  u:=\frac{A}{K-1}.
\]
Choose real numbers
\[
  1<c<s<u.
\]
For
\[
  \mathbf v=(v_2,\ldots,v_r)\in\Z_{\ge0}^{r-1},
\]
define
\[
  \omega(\mathbf v):=\sum_{j=2}^{r}(j-1)v_j,
  \qquad
  \abs{\mathbf v}:=\sum_{j=2}^{r}v_j,
  \qquad
  Y^{\mathbf v}:=\prod_{j=2}^{r}Y_j^{v_j}.
\]
Set
\begin{equation}\label{v:eq:W}
  W:=\floor{\frac{crm}{\log(er)}}
\end{equation}
and
\begin{equation}\label{v:eq:C}
  \cT:=\left\{
    \mathbf v\in\Z_{\ge0}^{r-1}:
    \omega(\mathbf v)\le W,
    \ \abs{\mathbf v}\le\ceil{sm}
  \right\}.
\end{equation}
Finally put
\begin{equation}\label{v:eq:B}
  B:=\ceil{\frac{mA}{K-1}}
\end{equation}
and let $\cQ$ be the span of all monomials
\begin{equation}\label{v:eq:Q-monomial}
  X^aY_0^{b_0}Y_1^{b_1}Y^{\mathbf v}
\end{equation}
with
\begin{align}
  &a,b_0,b_1\in\Z_{\ge0},\quad \mathbf v\in\cT,
  \quad b_1\le m,\label{v:eq:Q1}\\
  &b_0+b_1+\abs{\mathbf v}\le B,\label{v:eq:Q2}\\
  &a+(K-1)(b_0+b_1+\abs{\mathbf v})<mA.\label{v:eq:Q3}
\end{align}

\begin{lemma}[Degree properties]\label{v:lem:degree}
Every $Q\in\cQ$ satisfies $\deg_{Y_j}Q\le B$ for every $j$.  Moreover,
\[
  \deg_{(1,K-1,K-2,\ldots,K-r-1)}Q<mA.
\]
For every $P$ with $\deg P<K$,
\[
  \deg_X Q(X,P,P^{[1]},\ldots,P^{[r]})<mA.
\]
\end{lemma}

\begin{proof}
The individual-degree statement follows from \cref{v:eq:Q2}.  The actual differential-equation weights are no larger than the conservative weight $K-1$ used for every $Y_j$ in \cref{v:eq:Q3}.  Finally, $\deg P^{[j]}\le K-1$ for all $j$, so \cref{v:eq:Q3} also bounds the degree after specialization.
\end{proof}

\subsection{Local constraints and root multiplicity}\label{v:sec:local-constraints}

The decoder knows only the value $P(\alpha)$, not the Hasse derivatives of $P$ at $\alpha$.  The following substitution, constraints, and multiplicity argument are parameterized versions of the hidden-derivative construction of Brakensiek, Chen, Putterman, Zhang, and Zheng~\cite[Lemma~3.1]{BCPZZ26}.

Fix a point--symbol pair $(\alpha,y)\in F^2$.  In $Q$, make the formal substitution
\begin{equation}\label{v:eq:local-substitution}
  X=\alpha+T,
  \qquad
  Y_0=y+\sum_{j=1}^{r}(-1)^{j+1}T^jY_j+TE.
\end{equation}
Write the result uniquely as
\begin{equation}\label{v:eq:local-expansion}
  Q_{\alpha,y}(T,E,Y_1,\ldots,Y_r)
  =\sum_{b,\mathbf e}q_{b,\mathbf e}(T)
      E^bY_1^{e_1}\cdots Y_r^{e_r}.
\end{equation}
For every $b$ with $rb<m$, impose
\begin{equation}\label{v:eq:local-constraints}
  q_{b,\mathbf e}(T)\equiv0\pmod{T^{m-rb}}
  \qquad\text{for every }\mathbf e.
\end{equation}
These are homogeneous linear constraints on the coefficients of $Q$.

\begin{lemma}[One agreement gives multiplicity $m$]\label{v:lem:multiplicity}
If $P(\alpha)=y$ and $Q$ obeys \cref{v:eq:local-constraints} at $(\alpha,y)$, then
\[
 Q\bigl(\alpha+T,P(\alpha+T),P^{[1]}(\alpha+T),\ldots,
 P^{[r]}(\alpha+T)\bigr)
 \equiv0\pmod{T^m}.
\]
\end{lemma}

\begin{proof}
By \cref{v:eq:truncated-backward},
\[
 E_P(T):=
 \frac{P(\alpha+T)-y-
 \sum_{j=1}^{r}(-1)^{j+1}T^jP^{[j]}(\alpha+T)}{T}
\]
is divisible by $T^r$.  Substitute $Y_j=P^{[j]}(\alpha+T)$ and $E=E_P(T)$ into \cref{v:eq:local-expansion}.  If $rb<m$, the two factors $q_{b,\mathbf e}(T)$ and $E_P(T)^b$ contribute divisibility $T^{m-rb}$ and $T^{rb}$, respectively.  If $rb\ge m$, the second factor alone is divisible by $T^m$.  Thus every summand is divisible by $T^m$.
\end{proof}

\subsection{Dimension of the interpolation space}\label{v:sec:dimension}

The first side of the dimension--rank comparison is a lower bound on the number of admissible interpolation monomials.  The following box-counting estimate is a parameterized form of~\cite[Lemma~3.2]{BCPZZ26}: the ordinary-degree buffer $s<u$ leaves a three-dimensional box in the exponents of $X,Y_0,Y_1$ for every derivative-tail monomial.

\begin{lemma}[Dimension lower bound]\label{v:lem:dimension}
Fix $u_0>1$ and $1<s<u_0$, and put
\[
  \Delta:=u_0-s,
  \qquad
  \overline\Delta:=\min\{1,\Delta\}.
\]
Whenever
\[
  u=\frac{A}{K-1}\ge u_0,
  \qquad \overline\Delta m\ge2,
\]
we have
\begin{equation}\label{v:eq:dimension}
  \dim\cQ\ge
  \frac{\overline\Delta^3}{8192}
  \abs{\cT}(K-1)m^3.
\end{equation}
In particular, for fixed constants $u_0>s>1$ this gives
$\dim\cQ\ge a_0\abs{\cT}(K-1)m^3$ for all sufficiently large $m$, with
$a_0=\Omega(\min\{1,u_0-s\}^3)$.
\end{lemma}

\begin{proof}
Set
\[
  \tau:=\frac{\overline\Delta}{16}.
\]
Every $\mathbf v\in\cT$ satisfies
\[
  um-\abs{\mathbf v}
  \ge u_0m-sm-1
  =\Delta m-1.
\]
Choose arbitrary integers
\[
  0\le b_0,b_1\le\floor{\tau m}.
\]
Since $b_0+b_1\le2\tau m$ and $\overline\Delta m\ge2$,
\begin{align}
  um-\abs{\mathbf v}-b_0-b_1
  &\ge \Delta m-1-2\tau m \notag\\
  &\ge4\tau m.\label{v:eq:dim-margin}
\end{align}
Therefore every integer
\[
  0\le a<4\tau(K-1)m
\]
satisfies \cref{v:eq:Q3}.  Since $a\ge0$, that strict inequality also implies
\[
  b_0+b_1+\abs{\mathbf v}<\frac{mA}{K-1}\le B,
\]
so \cref{v:eq:Q2} holds; and $b_1\le m$ because $\tau\le1/16$.

Each $b_i$ interval contains at least $\tau m/2$ integers, while the $a$ interval contains at least $2\tau(K-1)m$ integers.  Hence each $\mathbf v\in\cT$ contributes at least
\[
  \frac{\tau m}{2}\cdot\frac{\tau m}{2}
  \cdot 2\tau(K-1)m
  =\frac{\tau^3}{2}(K-1)m^3
  =\frac{\overline\Delta^3}{8192}(K-1)m^3
\]
monomials, proving \cref{v:eq:dimension}.
\end{proof}

\subsection{Tail-monomial geometry}\label{v:sec:tail-geometry}

The local rank estimate depends on how much the set of derivative-tail monomials grows when its weighted cutoff is enlarged by $m$.  The simplex--cube comparison and centroid--Markov argument below parameterize~\cite[Corollary~3.8 and the proof of Lemma~3.3]{BCPZZ26}.

For $z\ge0$, define
\begin{equation}\label{v:eq:Lambda}
  \Lambda(z):=
  \abs{\{\mathbf v\in\Z_{\ge0}^{r-1}:\omega(\mathbf v)\le z\}}.
\end{equation}
Let
\[
  \mathsf S_z:=\left\{x\in\mathbb R_{\ge0}^{r-1}:
     \sum_{j=1}^{r-1}j x_j\le z\right\}.
\]
A change of variables from the standard simplex gives
\begin{equation}\label{v:eq:volume}
  \operatorname{Vol}(\mathsf S_z)
  =\frac{z^{r-1}}{((r-1)!)^2}.
\end{equation}
The usual half-open unit-cube comparison gives
\begin{equation}\label{v:eq:lattice-volume}
  \frac{z^{r-1}}{((r-1)!)^2}
  \le\Lambda(z)
  \le\frac{(z+r(r-1)/2)^{r-1}}{((r-1)!)^2}.
\end{equation}
Indeed, flooring coordinates maps $\mathsf S_z$ into the union of cubes based at counted lattice points, and every such cube lies in $\mathsf S_{z+r(r-1)/2}$.

\begin{lemma}[Uniform parameterized growth bound]\label{v:lem:growth}
Let $J\subset(1,\infty)$ be compact.  There are constants $C_J>0$ and $r_J$ such that, for every $r\ge r_J$, every $c\in J$, and every $s>c$, upon setting $m=r^3$ we have
\begin{equation}\label{v:eq:growth}
  \Lambda(W+m)
  \le \frac{C_J}{1-c/s}\abs{\cT}\,r^{1/c}.
\end{equation}
\end{lemma}

\begin{proof}
Let $X$ be uniform on $\mathsf S_W$.  The centroid of this simplex gives
\[
  \mathbb E\abs{X}_1
  =\frac{W}{r}\sum_{j=1}^{r-1}\frac1j
  \le\frac{W\log(er)}{r}
  \le cm.
\]
Markov's inequality yields
\[
  \Pr[\abs{X}_1>sm]\le\frac cs.
\]
Taking coordinatewise floors maps the remaining subset of $\mathsf S_W$ into cubes based at points of $\cT$, and the cubes are disjoint.  Hence
\begin{equation}\label{v:eq:C-volume-lower}
  \abs{\cT}\ge
  \left(1-\frac cs\right)
  \frac{W^{r-1}}{((r-1)!)^2}.
\end{equation}
Using the upper bound in \cref{v:eq:lattice-volume} and dividing by \cref{v:eq:C-volume-lower},
\begin{equation}\label{v:eq:ratio}
 \frac{\Lambda(W+m)}{\abs{\cT}}
 \le\frac{1}{1-c/s}
 \left(1+\frac{m+r(r-1)/2}{W}\right)^{r-1}.
\end{equation}
Since $m=r^3$ and
\[
  W=\frac{cr^4}{\log(er)}+O(1),
\]
uniformly for $c\in J$, after increasing $r_J$ if necessary we have
\[
  (r-1)\frac{m+r(r-1)/2}{W}
  \le\frac1c\log r+C_J'
\]
for a constant $C_J'$ depending only on $J$.  Apply $\log(1+x)\le x$ in \cref{v:eq:ratio} and exponentiate, absorbing $e^{C_J'}$ into $C_J$.
\end{proof}

\subsection{The local rank bound}\label{v:sec:local-rank}

We next bound the number of independent equations imposed by one point--symbol incidence.  The factorization through $\Gamma$, the kernel spaces, and the rank--nullity argument below are parameterized versions of~\cite[Lemmas~3.10--3.12]{BCPZZ26}.  The key is not to count all truncated coefficients directly: after a triangular change of variables, a large explicit subspace maps to zero.

Fix $(\alpha,y)$.  Let $\Phi_{\alpha,y}$ be the linear map that records exactly the coefficients prohibited by \cref{v:eq:local-constraints}.  First substitute
\begin{equation}\label{v:eq:first-sub}
  X=\alpha+T,
  \qquad Y_0=y+TU
\end{equation}
and reduce modulo $T^m$.  We regard $\cV$ and all $\cK_\rho$ below as subspaces of $F[T,U,Y_1,\ldots,Y_r]/(T^m)$; displayed generators denote residue classes, with the same convention after rewriting $U$ in terms of $E$.  The image lies in
\begin{equation}\label{v:eq:V}
 \cV:=\Span\left\{
 T^\rho U^aY_1^bY^{\mathbf v}:
 \begin{array}{l}
 0\le\rho<m,\ 0\le a\le\rho,\ 0\le b\le m,\\
 \omega(\mathbf v)\le W+\rho
 \end{array}
 \right\}.
\end{equation}
Now rewrite
\begin{equation}\label{v:eq:E}
 E=U-\sum_{j=1}^{r}(-1)^{j+1}T^{j-1}Y_j.
\end{equation}
Let $\Gamma$ record, after this rewrite, all coefficients of monomials
\[
 T^iE^bY_1^{e_1}Y^{\mathbf v}
 \qquad\text{with }i+rb<m.
\]
Then $\Phi_{\alpha,y}$ factors through $\Gamma$, so
\begin{equation}\label{v:eq:rank-factor}
  \rank\Phi_{\alpha,y}\le\rank\Gamma.
\end{equation}

For $0\le\rho<m$, set
\begin{equation}\label{v:eq:h}
  h_\rho:=\ceil{\frac{m-\rho}{r}}.
\end{equation}
If $\rho-h_\rho<0$, let $\cK_\rho=0$.  Otherwise, let $\cK_\rho$ be spanned by
\begin{equation}\label{v:eq:kernel-vectors}
  T^\rho E^{h_\rho}U^aY_1^bY^{\mathbf v}
\end{equation}
with
\[
  0\le a\le\rho-h_\rho,
  \quad 0\le b\le m-h_\rho,
  \quad \omega(\mathbf v)\le W+\rho.
\]

\begin{lemma}[Kernel spaces]\label{v:lem:kernel}
For every $\rho$, $\cK_\rho\subseteq\cV\cap\ker\Gamma$.  The spaces $\cK_0,\ldots,\cK_{m-1}$ are linearly independent, and
\begin{equation}\label{v:eq:kernel-dim}
 \dim\ker\Gamma\ge
 \sum_{\rho=0}^{m-1}
 \max\left\{0,
 (\rho-h_\rho+1)(m-h_\rho+1)\Lambda(W+\rho)
 \right\}.
\end{equation}
\end{lemma}

\begin{proof}
Expand $E^{h_\rho}$ by \cref{v:eq:E}.  If a term selects $e_j$ copies of $Y_j$ for $j\ge2$, it gains both $T$-degree and $\omega$-weight
\[
  L:=\sum_{j=2}^{r}(j-1)e_j.
\]
The resulting monomial has $T$-degree $\rho+L$, $U$-degree at most $h_\rho+a\le\rho\le\rho+L$, $Y_1$-degree at most $h_\rho+b\le m$, and $\omega$-weight at most $W+\rho+L$.  Every surviving term modulo $T^m$ therefore belongs to $\cV$.

After rewriting $U$ in terms of $E,Y_1,\ldots,Y_r$, every term in \cref{v:eq:kernel-vectors} remains divisible by $T^\rho E^{h_\rho}$.  Since $\rho+rh_\rho\ge m$, no resulting monomial is recorded by $\Gamma$.  Thus $\cK_\rho\subseteq\ker\Gamma$.

For fixed $\rho$, the displayed vectors are independent: their coefficient at the lowest $T$-degree is the common nonzero factor $(U-Y_1)^{h_\rho}$ times distinct monomials in $U,Y_1,Y_2,\ldots,Y_r$.  Spaces with different $\rho$ are independent by considering the smallest nonzero $T$-degree in a putative relation.  Counting the allowed exponents gives \cref{v:eq:kernel-dim}.
\end{proof}

\begin{lemma}[Local rank]\label{v:lem:local-rank}
Fix $1<c<s$, set $m=r^3$, and define $W,\cT$ as above.  For all sufficiently large $r$,
\begin{equation}\label{v:eq:local-rank}
  \rank\Phi_{\alpha,y}
  \le C(c,s)\abs{\cT}m^3r^{-(1-1/c)}.
\end{equation}
The bound is uniform in $(\alpha,y)$.
\end{lemma}

\begin{proof}
At fixed $T$-degree $\rho$, the space $\cV$ has dimension
\[
  (\rho+1)(m+1)\Lambda(W+\rho).
\]
By \cref{v:lem:kernel} and rank--nullity,
\begin{equation}\label{v:eq:rank-sum}
 \rank\Gamma\le
 \sum_{\rho=0}^{m-1}B_\rho\Lambda(W+\rho),
\end{equation}
where
\[
 B_\rho:=(\rho+1)(m+1)-
 \max\{0,(\rho-h_\rho+1)(m-h_\rho+1)\}.
\]
A direct expansion, including the case in which the maximum is zero, gives
\[
 B_\rho\le h_\rho(\rho+m+2),
 \qquad
 h_\rho\le\frac{m-\rho}{r}+1.
\]
Because $m=r^3$,
\begin{align*}
 \sum_{\rho=0}^{m-1}B_\rho
 &\le\sum_{\rho=0}^{m-1}
 \left(\frac{m-\rho}{r}+1\right)(\rho+m+2)\\
 &\le 6\frac{m^3}{r}
\end{align*}
for all sufficiently large $r$.  Apply \cref{v:lem:growth} with any compact interval $J\subset(1,\infty)$ containing $c$, and use monotonicity of $\Lambda$ together with \cref{v:eq:rank-factor,v:eq:rank-sum}:
\[
 \rank\Phi_{\alpha,y}
 \le6\frac{m^3}{r}\Lambda(W+m)
 \le C(c,s)\abs{\cT}m^3r^{1/c-1}.
\]
\end{proof}

\subsection{The global interpolation theorem}\label{v:sec:global-interpolation}

We now compare the dimension of $\cQ$ with the sum of the local ranks.  The first proposition records the derivative-order dependence established above by reparameterizing the dimension and rank arguments of~\cite{BCPZZ26} and records a fixed-shape version that makes the basic mechanism transparent.  The second lets the shape approach the full degree slack and is the form used by the decoder.

We first give a clean fixed-shape formulation.  Fix $\vartheta,\lambda\in(0,1)$ and define
\begin{equation}\label{v:eq:fixed-shape}
 c:=\frac{1}{1-\lambda\vartheta},
 \qquad
 s:=\frac12\left(
  \frac{1}{1-\lambda\vartheta}
  +\frac{1}{1-\vartheta}
 \right).
\end{equation}
Then
\[
 1<c<s<\frac{1}{1-\vartheta},
 \qquad
 1-\frac1c=\lambda\vartheta.
\]

\begin{proposition}[Fixed-shape sparse interpolant]\label{v:prop:fixed-interpolant}
Fix $\vartheta,\lambda\in(0,1)$.  There are effective constants $r_0$ and $a_{\vartheta,\lambda}>0$ such that the following holds.  Suppose
\[
 \frac{A}{K-1}\ge\frac{1}{1-\vartheta},
 \qquad r\ge r_0,
 \qquad m=r^3.
\]
For a list-recovery instance of total incidence $M$, if
\begin{equation}\label{v:eq:fixed-incidence}
  M\le a_{\vartheta,\lambda}(K-1)r^{\lambda\vartheta},
\end{equation}
then there is a nonzero $Q\in\cQ$ satisfying the local constraints at every incidence
\[
  (\alpha_i,y),\qquad y\in S_i.
\]
Every degree-below-$K$ polynomial hitting at least $A$ lists satisfies
\begin{equation}\label{v:eq:diffeq}
  Q(X,P,P^{[1]},\ldots,P^{[r]})\equiv0.
\end{equation}
\end{proposition}

\begin{proof}
Apply \cref{v:lem:dimension} with $u_0=1/(1-\vartheta)$ and the value of $s$ in \cref{v:eq:fixed-shape}, increasing $r_0$ so that $\min\{1,u_0-s\}r^3\ge2$.  It gives
\[
  \dim\cQ\ge a_0\abs{\cT}(K-1)m^3.
\]
By \cref{v:lem:local-rank}, one incidence has rank at most
\[
  C_0\abs{\cT}m^3r^{-\lambda\vartheta}.
\]
The stacked map over all $M$ incidences has rank at most the sum of these ranks.  Taking $a_{\vartheta,\lambda}<a_0/C_0$ leaves a nonzero vector in its kernel.

Let $P$ hit the lists on a set $I$ with $\abs I\ge A$.  For every $i\in I$, apply \cref{v:lem:multiplicity} with $y=P(\alpha_i)\in S_i$.  The specialization in \cref{v:eq:diffeq} has $A$ distinct roots, each of multiplicity at least $m$.  By \cref{v:lem:degree}, its degree is below $mA$, so it is identically zero.
\end{proof}

\subsubsection{Using the full slack}

We now track the constants while allowing the shape to depend mildly on $r$.

\begin{proposition}[Optimized sparse interpolant]\label{v:prop:optimized-interpolant}
Fix $\vartheta\in(0,1)$.  There are effective constants $r_0$ and $a_\vartheta>0$ such that, for integers $2\le K\le A\le n$ and $1\le r<K$, whenever
\[
  \frac{A}{K-1}\ge u_0:=\frac{1}{1-\vartheta},
  \qquad r\ge r_0,
  \qquad m=r^3,
\]
and
\begin{equation}\label{v:eq:optimized-incidence}
  M\le a_\vartheta(K-1)
       \frac{r^\vartheta}{\log^4(er)},
\end{equation}
there is a nonzero interpolant satisfying every incidence constraint, and every degree-below-$K$ polynomial with at least $A$ list hits satisfies its differential equation.
\end{proposition}

\begin{proof}
Put $L=\log(er)$ and, for sufficiently large $r$, choose
\begin{equation}\label{v:eq:optimized-shape}
  \delta:=\frac1L,
  \qquad
  c:=u_0-2\delta,
  \qquad
  s:=u_0-\delta.
\end{equation}
Increasing $r_0(\vartheta)$ ensures $1<c<s<u_0$, that
\[
  c\in J_\vartheta:=\left[\frac{1+u_0}{2},u_0\right],
\]
and that $\delta m=r^3/\log(er)\ge2$.

Apply \cref{v:lem:dimension} with gap $u_0-s=\delta$ to obtain
\begin{equation}\label{v:eq:optimized-dimension}
  \dim\cQ\ge c_\vartheta\delta^3
  \abs{\cT}(K-1)m^3.
\end{equation}
Apply the uniform form of \cref{v:lem:growth} with the compact interval $J_\vartheta$.  Since
\[
  \frac{1}{1-c/s}=\frac{s}{s-c}=O_\vartheta(\delta^{-1}),
\]
we obtain
\[
  \Lambda(W+m)
  \le C_\vartheta\delta^{-1}\abs{\cT}r^{1/c}.
\]
Moreover,
\[
  \vartheta-\left(1-\frac1c\right)
  =\frac1c-\frac1{u_0}
  =\frac{2\delta}{cu_0}.
\]
Since $\delta\log r=O(1)$,
\[
  r^{1/c-1}=O_\vartheta(r^{-\vartheta}).
\]
Combining this uniform estimate with the bound $\rank\Phi_{\alpha,y}\le6m^3\Lambda(W+m)/r$ established in the proof of \cref{v:lem:local-rank} gives
\begin{equation}\label{v:eq:optimized-rank}
  \rank\Phi_{\alpha,y}
  \le C_\vartheta\delta^{-1}
  \abs{\cT}m^3r^{-\vartheta}.
\end{equation}
The total rank over all incidences is smaller than \cref{v:eq:optimized-dimension} whenever
\[
  M\le c_\vartheta'(K-1)\delta^4r^\vartheta,
\]
which is \cref{v:eq:optimized-incidence}.  The multiplicity argument is unchanged.
\end{proof}

\subsection{The coefficient-level interpolation interface}\label{sec:matrix-interface}
For symbolic applications we need the linear system itself, not just the existence of an interpolant for each received word.

\begin{corollary}[Specialization-compatible interpolation]\label{cor:interpolation-interface}
Fix $\vartheta\in(0,1)$, and let $r_0,a_\vartheta$ be the constants of \cref{v:prop:optimized-interpolant}. Let $F$ be any field, let $\alpha_1,\ldots,\alpha_n\in F$ be distinct, and let $S_i\subseteq F$ be finite sets of total incidence $M$. Suppose
\[
 2\le K\le A\le n,\qquad r_0\le r<K,\qquad m=r^3,
\]
\[
 \frac{A}{K-1}\ge\frac1{1-\vartheta},\qquad
 M\le a_\vartheta(K-1)\frac{r^\vartheta}{\log^4(er)}.
\]
Put $B=\ceil{mA/(K-1)}$. There is a homogeneous interpolation matrix, in the fixed monomial basis of \cref{v:eq:Q-monomial,v:eq:Q1,v:eq:Q2,v:eq:Q3}, with a nonzero kernel. Every nonzero kernel vector gives a polynomial $Q$ satisfying
\begin{equation}\label{eq:interface-degrees}
 \deg_Y^{\rm tot}Q\le B,\qquad
 \deg_{(1,K-1,\ldots,K-1)}Q<mA.
\end{equation}
Every degree-below-$K$ polynomial with at least $A$ list hits satisfies
\begin{equation}\label{eq:interface-identity}
 Q(X,P,P^{[1]},\ldots,P^{[r]})\equiv0.
\end{equation}
For fixed $r,B$, the matrix has $O_{r,B}(n)$ columns and $O_{r,B}(M)$ rows. Each entry in the block for $(\alpha_i,y)$ is a polynomial of degree at most $B$ in $y$. All these assertions about a kernel vector and its local constraints remain valid under extension of the coefficient field.
\end{corollary}
\begin{proof}
The matrix records exactly the coefficient constraints in \cref{v:eq:local-constraints}. The nonzero kernel is supplied by \cref{v:prop:optimized-interpolant}. The fixed basis enforces \cref{eq:interface-degrees} for every vector, and the multiplicity argument in \cref{v:lem:multiplicity} proves \cref{eq:interface-identity} for every kernel vector.

For fixed $r,B$, there are constantly many possible $Y$-monomials and fewer than $mA=O_r(n)$ choices of the $X$-exponent. In a local block, both the truncated $T$-degree and every exponent in the substituted derivative variables and $E$ are bounded in terms of $r,B$. Thus each incidence contributes only constantly many rows. In the substitution for $Y_0$, the received symbol occurs to degree at most its original exponent, which is at most $B$. These are identities over the coefficient field, so extension of scalars preserves the local equations and the degree--multiplicity implication.
\end{proof}

\section{Uniform parameters and deterministic decoding}\label{sec:decoding}
\subsection{Padding at the actual rate}
The following choice is shared by the algorithmic and geometric arguments. It makes the constants independent of the actual rate $k/n$.

\begin{lemma}[Uniform fixed-slack parameters]\label{lem:uniform-padding}
Fix $\gamma\in(0,1)$ and an integer $L\ge1$. For all sufficiently large $n$, and every pair of integers
\[
 1\le k\le(1-\gamma)n,\qquad k+\ceil{\gamma n}\le A\le n,
\]
one can choose $r=r(\gamma,L)$ and $m=r^3$, independently of $n,k,A$, so that, with
\begin{equation}\label{w:eq:A-K}
 \vartheta:=\frac\gamma2,\qquad K:=\ceil{(1-\vartheta)A},
\end{equation}
all hypotheses of \cref{cor:interpolation-interface} hold for every instance with total incidence $M\le Ln$. Moreover,
\begin{equation}\label{w:eq:B-constant}
 K-k\ge\frac{\gamma n}{2},\qquad
 B=\ceil{\frac{mA}{K-1}}\le B_{\gamma,L}
\end{equation}
for an effective constant $B_{\gamma,L}$. For every prime $q\ge n$, after increasing the same lower threshold for $n$, one also has
\[
 r<K\le q,\qquad B<q,\qquad r^3A<q^2.
\]
\end{lemma}
\begin{proof}
Choose an integer $r\ge r_0(\vartheta)$ sufficiently large that
\begin{equation}\label{w:eq:r-choice}
 1\le \frac{a_\vartheta(1-\vartheta)\gamma}{3L}
             \frac{r^\vartheta}{\log^4(er)}.
\end{equation}
Such an effective choice exists because $r^\vartheta/\log^4(er)$ tends to infinity. Since $A\ge k+\gamma n$ and $k+\gamma n\le n$,
\begin{equation}\label{w:eq:K-ge-k}
 K-k\ge(1-\vartheta)(k+\gamma n)-k
      =\gamma n-\vartheta(k+\gamma n)\ge\gamma n/2.
\end{equation}
For sufficiently large $n$, $2\le K\le A$ and $r<K$. Also
\[
 K-1<(1-\vartheta)A,\qquad
 K-1\ge\frac{(1-\vartheta)\gamma n}{2}.
\]
These inequalities give the required ratio and, with \cref{w:eq:r-choice}, the incidence budget for $M\le Ln$. For large $n$, the stronger estimate $K-1\ge(1-\vartheta)A/2$ gives $B\le\ceil{2m/(1-\vartheta)}$. Finally $B$ and $m$ are fixed: taking $n>B$ and $n>m$ gives $B<q$ and $mA\le mn<n^2\le q^2$.
\end{proof}

\subsection{Prime-field differential-equation root finding}

We use the prime-field formulation in~\cite[Theorem~2.1]{BCPZZ26}, which restates Kopparty's differential-equation root solver~\cite[Theorem~4.3]{Kop15}.

\begin{theorem}[Kopparty; prime-field form from BCPZZ]\label{v:thm:kopparty-prime}
Let $q$ be prime and $q\ge K>r$.  Suppose
\[
  0\ne Q\in\F_q[X,Y_0,\ldots,Y_r]
\]
satisfies $\deg_{Y_j}Q<q$ for all $j$ and
\[
  \deg_{(1,K-1,K-2,\ldots,K-r-1)}Q<q^2.
\]
Then all $P\in\F_q[X]$ with $\deg P<K$ satisfying
\[
  Q(X,P,P^{[1]},\ldots,P^{[r]})\equiv0
\]
can be found in $q^{O(r+1)}$ time.  Their number is at most $q^{4r+6}$.
\end{theorem}

\begin{remark}[The number of algebraic roots versus the final list]
The root bound in \cref{v:thm:kopparty-prime} includes every degree-bounded solution of the differential equation, whether or not it has the required agreement. In \cref{sec:list-bounds} we bound the much smaller set that passes the agreement test. From Kopparty's bound $t(r+1)q^{2r\lfloor d/q\rfloor+4r+4}$ with $d=K-1<q$, $t<q$, and $r+1\le q$, one obtains the displayed $q^{4r+6}$ estimate~\cite[Theorem~4.3]{Kop15}.
\end{remark}
\subsection{Sparse list recovery and its running time}\label{v:sec:algorithm}

The interpolation theorem becomes an algorithm once it is combined with Kopparty's prime-field root-enumeration theorem.  We state the sparse-incidence form, from which ordinary list decoding and bounded-list recovery follow immediately.

\begin{theorem}[Prime-field sparse list recovery]\label{v:thm:prime-sparse}
Fix $\vartheta\in(0,1)$.  There are effective constants $r_0,a_\vartheta>0$ with the following property.

Let $q$ be prime, let $\alpha_1,\ldots,\alpha_n\in\F_q$ be distinct, and let $S_i\subseteq\F_q$ have total incidence $M$.  Let $A,k,K$ be integers with $1\le A\le n$, $1\le k\le K$, and $K\ge2$, and suppose
\begin{align}
 &\frac{A}{K-1}\ge\frac{1}{1-\vartheta},\label{v:eq:alg1}\\
 &r_0\le r<K,
 \qquad
 M\le a_\vartheta(K-1)\frac{r^\vartheta}{\log^4(er)},\label{v:eq:alg2}\\
 &B:=\ceil{\frac{r^3A}{K-1}}<q,
 \qquad r^3A<q^2.\label{v:eq:alg3}
\end{align}
Then a deterministic algorithm outputs every $P\in\F_q[X]$ with $\deg P<k$ and at least $A$ list hits.  It uses
\begin{equation}\label{v:eq:prime-runtime}
  q^{O_\vartheta(r)}
\end{equation}
field operations, and the output list has size at most $q^{4r+6}$.
\end{theorem}

\begin{proof}
Construct $\cQ$ using the optimized shape \cref{v:eq:optimized-shape}.  Form the homogeneous local constraints for every incidence and use Gaussian elimination to find the nonzero interpolant promised by \cref{v:prop:optimized-interpolant}.  By \cref{v:lem:degree,v:eq:alg3}, its individual $Y_j$-degrees are below $q$ and its differential weighted degree is below $q^2$.  Since distinct evaluation points imply $n\le q$ and \cref{v:eq:alg1} implies $K\le A\le n$, we also have $q\ge K>r$.  Apply \cref{v:thm:kopparty-prime}, then filter by degree $k-1$ and agreement.

It remains to bound interpolation time.  Since $s=O_\vartheta(1)$ and $m=r^3$,
\begin{equation}\label{v:eq:C-count}
  \abs{\cT}\le
  \binom{\ceil{sm}+r-1}{r-1}
  =r^{O_\vartheta(r)}.
\end{equation}
Condition $B<q$ and $B\ge r^3$ imply $q>r$.  Thus \cref{v:eq:C-count} is $q^{O_\vartheta(r)}$.  There are at most $q$ choices for each of $b_0,b_1$, and fewer than $q^2$ choices for $a$, so $\cQ$ has $q^{O_\vartheta(r)}$ monomials.

The coefficient records in one local map have
\[
  0\le i<m,
  \qquad 0\le b<m/r,
  \qquad e_1<2m,
  \qquad \omega(\mathbf v)\le W+m.
\]
Indeed, before the $U$-rewrite, a term of $T$-degree $\rho$ has $Y_1$-degree at most $m$ and $U$-degree at most $\rho<m$, so rewriting $U$ can increase the $Y_1$-degree only to $m+\rho<2m$.  Its initial tail weight is at most $W+\rho$, and each added $Y_j$ with $j\ge2$ increases both the $T$-degree and the $\omega$-weight by $j-1$.  Thus a final term of $T$-degree $i<m$ has tail weight at most $W+i<W+m$.  Their number is $r^{O_\vartheta(r)}=q^{O_\vartheta(r)}$.  There are at most $M\le nq\le q^2$ incidences.  Each matrix entry can be computed by truncated polynomial arithmetic in $q^{O_\vartheta(r)}$ operations.  Hence the whole matrix has $q^{O_\vartheta(r)}$ rows and columns and can be formed and row-reduced in the same asymptotic form.  Root finding and final verification also fit \cref{v:eq:prime-runtime}.
\end{proof}

\begin{remark}[Where the runtime improvement comes from]
The two-cutoff derivative-tail space already appears in~\cite{BCPZZ26}.  For the runtime bound we count its admissible tails directly as weak compositions under $\abs{\mathbf v}=O(r^3)$, obtaining only $r^{O(r)}$ possibilities.  This is a sharper enumeration of the existing interpolation space; no new linear-algebra primitive is required.
\end{remark}

\begin{corollary}[Algorithmic part of \cref{thm:capacity}]\label{cor:algorithmic-capacity}
Under the hypotheses of \cref{thm:capacity}, the desired list is computable in $q^{O_{\gamma,L}(1)}$ field operations. Before agreement filtering, the differential root solver produces at most $q^{4r+6}$ candidates, where $r=r(\gamma,L)$ is from \cref{lem:uniform-padding}.
\end{corollary}
\begin{proof}
Set $A=k+\ceil{\gamma n}$. The bound on $k$ ensures $A\le n$. Apply \cref{lem:uniform-padding}, use $M\le Ln$, and invoke \cref{v:thm:prime-sparse}. Every desired polynomial has degree below $K$ and is preserved by the final degree and agreement filters.
\end{proof}

For fixed $R,\gamma,L$ and $k\le Rn$, requiring $\ceil{(R+\gamma)n}$ list hits instead only restricts the output further. The constants above depend on $\gamma,L$ alone. A more quantitative derivative-order choice in terms of a multiplicative slack and a target agreement is given in \cref{app:quantitative}.

\section{Symbolic interpolation and solution geometry}\label{sec:geometry}
We now replace the finite-field enumeration of roots by a geometric description that is useful both for counting actual list candidates and for following their supports along curves. The symbolic construction is used for MCA; a parameter-free corollary of the same geometric theorem will control individual lists.

\subsection{A primitive polynomial interpolant over \texorpdfstring{$\F_q[Z]$}{Fq[Z]}}

For the curve in \cref{w:thm:mca}, put $A=n-\floor{\rho n}$ and use \cref{lem:uniform-padding} with $L=1$. Write $B_\gamma=B_{\gamma,1}$. Define symbolic received symbols
\begin{equation}\label{w:eq:symbolic-y}
 y_i(Z):=\sum_{j=0}^{\ell}f_{j,i}Z^j\in\F_q[Z].
\end{equation}
Apply \cref{cor:interpolation-interface} over $\F_q(Z)$.  We need a version whose coefficients are polynomial in $Z$ and that remains nonzero under every specialization.

\begin{lemma}[Polynomial kernel vector]\label{w:lem:polynomial-kernel}
Let $M(Z)\in F[Z]^{u\times v}$ have entries of degree at most $E$, and suppose $\rank_{F(Z)}M<v$.  Then $\ker_{F[Z]}M$ contains a nonzero vector $c(Z)$ whose coordinates have degree at most $E\min\{u,v-1\}$.  The coordinates may be chosen primitive, meaning that their greatest common divisor in $F[Z]$ is $1$.
\end{lemma}

\begin{proof}
Let $s=\rank_{F(Z)}M$.  If $s=0$, take a standard basis vector.  Otherwise choose a nonzero $s\times s$ minor and one additional column.  The signed $s\times s$ cofactors of the resulting $s\times(s+1)$ submatrix give a nonzero relation among the selected columns on the $s$ selected rows.  Those rows are linearly independent and hence form a basis of the row space over $F(Z)$, so the same relation annihilates every row of $M$.  Thus the cofactor vector, extended by zero on the other columns, lies in $\ker_{F[Z]}M$.  Every cofactor has degree at most $sE$.  Divide all coordinates by their common polynomial factor.  Because $F[Z]$ is a domain, the divided vector remains in the kernel.
\end{proof}

\begin{proposition}[Symbolic differential interpolant]\label{w:prop:symbolic-Q}
Under the hypotheses and notation of \cref{w:thm:mca,lem:uniform-padding}, there are constants $r=r(\gamma)$, $B=B(\gamma)$, and $c=c(\gamma,\ell)$, and a polynomial
\[
 Q(Z,X,Y_0,\ldots,Y_r)\in\F_q[Z,X,Y_0,\ldots,Y_r]
\]
with the following properties:
\begin{enumerate}[label=\textup{(\alph*)}]
 \item $Q$ is primitive as a polynomial in $Z$, i.e., the gcd in $\F_q[Z]$ of its coefficients in $X,Y_0,\ldots,Y_r$ is $1$;
 \item
 \[
  \deg_Z Q\le cn,
  \qquad \deg_X Q\le cn,
  \qquad \deg_Y^{\mathrm{tot}}Q\le B;
 \]
 \item for every $z\in\overline{\F_q}$, the specialization $Q_z:=Q(z,X,Y_0,\ldots,Y_r)$ is nonzero;
 \item for every $z\in\F_q$ and every $P\in\F_q[X]$ of degree below $k$ that agrees with $(y_i(z))_{i=1}^n$ on at least $A$ coordinates,
 \begin{equation}\label{w:eq:specialized-diffeq}
  Q\bigl(z,X,P,P^{[1]},\ldots,P^{[r]}\bigr)\equiv0.
 \end{equation}
\end{enumerate}
\end{proposition}

\begin{proof}
Use the interpolation matrix in \cref{cor:interpolation-interface}, replacing each received symbol by \cref{w:eq:symbolic-y}.  Since the matrix entries have degree at most $B$ in $y_i$, they have $Z$-degree at most $\ell B$.  The matrix has $O_{\gamma}(n)$ rows and columns.  By \cref{w:lem:polynomial-kernel}, it has a nonzero polynomial kernel vector of $Z$-degree $O_{\gamma,\ell}(n)$.  Interpret that vector as the coefficient vector of $Q$ in the hidden-derivative monomial basis, and divide its content in $\F_q[Z]$.  The divided vector remains in the kernel, so the matrix formulation in \cref{cor:interpolation-interface} guarantees all local interpolation constraints for this particular primitive choice.

The $X$- and $Y$-degree bounds follow from \cref{eq:interface-degrees} and $A\le n$.  Primitivity implies that no $z\in\overline{\F_q}$ can annihilate every coefficient of $Q$: otherwise the minimal polynomial of $z$ over $\F_q$ would divide their gcd.  This proves (c).

All interpolation constraints are polynomial identities in $Z$.  They therefore remain valid after specializing $Z=z$.  The coupled monomial bound in \cref{eq:interface-degrees} and the multiplicity argument in \cref{v:lem:multiplicity} then gives \cref{w:eq:specialized-diffeq} for every nearby $P$.
\end{proof}

\subsection{Cumulative degree}

All algebraic varieties in the proof are considered over an algebraic closure.  The degree of an affine irreducible variety means the degree of its projective closure.  If $V$ is an affine algebraic set with irreducible components $V_1,\ldots,V_t$, define its cumulative degree by
\[
 \cdeg(V):=\sum_{j=1}^{t}\deg(V_j).
\]
We use two standard consequences of affine/projective B\'ezout; see, for example,~\cite{Hei83,Har92}.

\begin{lemma}[Effective intersections]\label{w:lem:effective-intersection}
Let $V\subseteq\A^N$ have dimension at most $s$ and cumulative degree $\Delta$.  If $W\subseteq V$ is cut out inside $V$ by an arbitrary collection of polynomials of degree at most $E\ge1$, then
\[
 \cdeg(W)\le \Delta E^{s}.
\]
If $U=W\cap\{H\ne0\}$ is a principal-open subset of $W$, then $\overline U$ is a union of irreducible components of $W$, so the same bound holds for $\cdeg(\overline U)$.
\end{lemma}

\begin{proof}
By Noetherianity, finitely many of the equations cut out $W$ set-theoretically inside $V$.  Decompose $V$ into irreducible components and organize the successive intersections as a component tree, retaining an equation on a branch only when it does not vanish identically on the current component.  Every retained intersection strictly decreases dimension, so a branch has length at most $s$; at each retained step, affine/projective B\'ezout multiplies the sum of the degrees of the children by at most $E$.  Components on which later equations vanish identically are carried forward without a degree increase.  Padding shorter branches by the factor $E\ge1$ gives total leaf degree at most $\Delta E^s$, and every irreducible component of $W$ is one of these leaves.
\end{proof}

\begin{lemma}[Graph of a bounded-degree rational map]\label{w:lem:rational-graph}
Let $U\subseteq\A^m$ be constructible, with $\dim\overline U\le s$ and $\cdeg(\overline U)\le\Delta$.  Suppose a rational map $\phi:U\to\A^N$ has all coordinates of the form
\[
 \phi_j=\frac{N_j}{H^{e_j}},
\]
where $H$ is nonzero on $U$, $e_j,L$ are nonnegative integers with $e_j\le L$, and
\[
 \deg H,\ \deg N_j,\ L\deg H\le E.
\]
Then the Zariski closure of the graph of $\phi$ has dimension at most $s$ and cumulative degree at most
\[
 \Delta(2E+2)^s.
\]
In particular, when $s$ is fixed and $E,N$ are polynomially bounded in a parameter $D$, the graph has degree $D^{O_s(1)}$.
\end{lemma}

\begin{proof}
Put $G=H^L$.  On the open set $H\ne0$, the graph is the image of the graph map
\[
 \Psi(u):=
 [G(u):u_1G(u):\cdots:u_mG(u):
   N_1(u)H(u)^{L-e_1}:\cdots:N_N(u)H(u)^{L-e_N}].
\]
After passing to the projective closure of each irreducible component of $\overline U$ and homogenizing the displayed coordinates to a common degree, $\Psi$ is a rational projective map defined by forms of degree at most $2E+2$.  The source coordinates occur in the image, so $\Psi$ is generically one-to-one on every component meeting $U$.  The standard degree bound for a rational image~\cite{Har92} therefore gives, for a component $X$ of dimension $t\le s$,
\[
 \deg\overline{\Psi(X\cap U)}
 \le \deg(X)(2E+2)^t.
\]
Indeed, intersecting the image with a generic codimension-$t$ linear space and pulling back gives $t$ members of this degree-$(2E+2)$ linear system on $X$; after discarding components supported in the base locus, projective B\'ezout bounds the resulting moving zero-cycle by the right-hand side.
Summing over the irreducible components of $\overline U$ yields the stronger bound
\[
 \cdeg(\overline{\Gamma_\phi})
 \le \Delta(2E+2)^s,
\]
as asserted. The same generically injective parametrization gives $\dim\overline{\Gamma_\phi}\le s$.
\end{proof}

\subsection{Effective jet decomposition of the solution locus}

This section proves the algebraic statement that converts one symbolic differential equation into a bounded-complexity family of polynomial solutions. The nonsingular Taylor lifting and the singular reduction by a highest-variable partial derivative follow Kopparty~\cite[Section~4.2, Theorem~4.4 and the proof of Theorem~4.3]{Kop15}. Here we add polynomial numerator and denominator degree bounds, parameter-specialization control, and a bounded-cumulative-degree geometric cover.

Let $F$ be algebraically closed of characteristic $p$, let $d<p$, and write
\[
 P_{\mathbf p}(X)=\sum_{j=0}^{d}p_jX^j
\]
for the polynomial corresponding to $\mathbf p=(p_0,\ldots,p_d)\in F^{d+1}$.  For
\[
 Q\in F[Z,X,Y_0,\ldots,Y_r],
\]
define
\begin{equation}\label{w:eq:solution-locus}
 \Sol_d(Q):=
 \left\{(z,\mathbf p)\in\A^{d+2}:
 Q\bigl(z,X,P_{\mathbf p},P_{\mathbf p}^{[1]},\ldots,
 P_{\mathbf p}^{[r]}\bigr)\equiv0
 \right\}.
\end{equation}

A polynomial in $F[Z,X,Y_0,\ldots,Y_r]$ is \emph{$Z$-primitive} if the gcd of its coefficients as a polynomial in $X,Y_0,\ldots,Y_r$ is $1$ in $F[Z]$.

\begin{theorem}[Parametric jet decomposition]\label{w:thm:jet-decomposition}
Fix integers $r,B\ge0$.  There is an effective constant $C_0=C_0(r,B)$ with the following property.  Let $F$ be algebraically closed of characteristic $p$, let $0\le d<p$, and let
\[
 0\ne Q\in F[Z,X,Y_0,\ldots,Y_r]
\]
be $Z$-primitive and satisfy
\begin{equation}\label{w:eq:jet-degree-input}
 \deg_ZQ,\deg_XQ\le D,
 \qquad
 \deg_Y^{\mathrm{tot}}Q\le B<p,
 \qquad D\ge2.
\end{equation}
Then there are a finite set $E_Q\subseteq F$ and affine algebraic sets
\[
 V_1,\ldots,V_t\subseteq\A^{d+2}
\]
such that
\begin{align}
 &\abs{E_Q}+\sum_{j=1}^{t}\cdeg(V_j)
   \le (D+d+2)^{C_0},\label{w:eq:jet-complexity}\\
 &\dim V_j\le r+1\quad\text{for every }j,\label{w:eq:jet-dimension}
\end{align}
and every $(z,\mathbf p)\in\Sol_d(Q)$ with $z\notin E_Q$ belongs to $\bigcup_jV_j$.
\end{theorem}

The exceptional set is needed because, after differentiating $Q$ with respect to its highest derivative variable, all coefficients of that derivative may vanish at finitely many parameter values.  In our proximity application we select at most one polynomial above each $z$, so those fibers contribute only $\abs{E_Q}$ selected points.

\subsubsection{Triangular Hasse--Taylor lifting}

We first give the quantitative local lifting statement.  Suppose $Q$ depends on $Y_s$ but on no $Y_j$ with $j>s$.  Fix $\alpha\in F$ and write
\[
 P(\alpha+U)=\sum_{j=0}^{d}c_jU^j.
\]
Set
\begin{equation}\label{w:eq:Halpha}
 H_\alpha(Z,c_0,\ldots,c_s)
 :=\frac{\partial Q}{\partial Y_s}
       (Z,\alpha,c_0,\ldots,c_s).
\end{equation}
For $t\ge0$, let $F_{\alpha,t}$ be the coefficient of $U^t$ in
\begin{equation}\label{w:eq:Q-Taylor-sub}
 Q\left(Z,\alpha+U,
       P(\alpha+U),P^{[1]}(\alpha+U),\ldots,P^{[s]}(\alpha+U)
 \right).
\end{equation}

\begin{lemma}[Triangular coefficient equation]\label{w:lem:triangular}
For $1\le t\le d-s$,
\begin{equation}\label{w:eq:triangular-equation}
 F_{\alpha,t}
 =\binom{s+t}{s}H_\alpha(Z,c_0,\ldots,c_s)c_{s+t}
  +G_{\alpha,t}(Z,c_0,\ldots,c_{s+t-1})
\end{equation}
for a polynomial $G_{\alpha,t}$.  In particular, because $d<p$, the coefficient $\binom{s+t}{s}$ is nonzero in $F$.
\end{lemma}

\begin{proof}
Perturb $c_{s+t}$ while fixing all earlier Taylor coefficients.  By \cref{w:eq:hasse-series}, the induced perturbation of $P^{[s]}$ begins with
\[
 \binom{s+t}{s}c_{s+t}U^t,
\]
whereas the perturbation of $P^{[i]}$ for $i<s$ begins in degree $s+t-i>t$.  In the coefficient of $U^t$ in \cref{w:eq:Q-Taylor-sub}, only the term linear in the perturbation of $Y_s$ can occur; two such perturbations have order at least $2t>t$.  Its coefficient is the constant term of $\partial Q/\partial Y_s$ after substitution, namely $H_\alpha$.  This proves \cref{w:eq:triangular-equation}.  Since $s+t\le d<p$, the displayed binomial coefficient is nonzero.
\end{proof}

\begin{lemma}[Polynomial degree of the lifted jet]\label{w:lem:rational-jet-degree}
Assume \cref{w:eq:jet-degree-input}, and suppose the total $Y$-degree is at most $B$.  On the open set $H_\alpha\ne0$, the equations
\[
 F_{\alpha,t}=0,
 \qquad 1\le t\le d-s,
\]
recursively determine
\begin{equation}\label{w:eq:rational-coefficients}
 c_{s+t}=\frac{N_t(Z,c_0,\ldots,c_s)}
                 {H_\alpha(Z,c_0,\ldots,c_s)^{e_t}}.
\end{equation}
There is a constant $C_B$ depending only on $B$ such that
\begin{equation}\label{w:eq:rational-degree-bound}
 e_t\le 2t-1,
 \qquad
 \deg N_t\le C_B(D+B)t^2.
\end{equation}
After these substitutions, every remaining coefficient equation
$F_{\alpha,t}=0$ has a numerator of degree at most
\begin{equation}\label{w:eq:residual-degree}
 E_0:=C_B(D+B)(d+2)^2.
\end{equation}
\end{lemma}

\begin{proof}
The recursive existence follows from \cref{w:lem:triangular}.  We prove the denominator bound and the stronger numerator estimate $\deg N_t\le4(D+B)t^2$ by induction on $t$.  Consider a monomial contributing to the coefficient of $U^t$ in \cref{w:eq:Q-Taylor-sub}.  It contains at most $B$ Taylor coefficients.  If $c_j$ with $j>s$ occurs through $P^{[i]}$, its $U$-order is $j-i\ge j-s$.  Hence, if the occurrences above the initial jet are $c_{s+h_1},\ldots,c_{s+h_v}$, then
\begin{equation}\label{w:eq:order-budget}
 h_1+\cdots+h_v\le t,
 \qquad v\le B.
\end{equation}
In the remainder $G_{\alpha,t}$ from \cref{w:eq:triangular-equation}, either $h_1+\cdots+h_v\le t-1$, or equality holds and $v\ge2$: the only term with $v=1$ and $h_1=t$ is the linear term in $c_{s+t}$ already separated in \cref{w:eq:triangular-equation}.  Using the inductive bound $e_h\le2h-1$, every monomial of $G_{\alpha,t}$ therefore has denominator exponent at most
\[
 \sum_{j=1}^{v}(2h_j-1)\le2t-2.
\]
After putting $G_{\alpha,t}$ over the common denominator $H_\alpha^{2t-2}$, solving \cref{w:eq:triangular-equation} for $c_{s+t}$ gives $e_t\le2t-1$.

For the numerator degree, the same separation gives
\[
 \sum_{j=1}^{v}h_j^2\le t^2-2t+2
\]
whenever $t\ge2$: under a fixed sum, the sum of squares is maximized by concentrating all but one unit in a single part, and the one-part case of total weight $t$ is absent from $G_{\alpha,t}$.  By induction, the numerator of a substituted monomial therefore has degree at most
\[
 D+B+4(D+B)(t^2-2t+2).
\]
Passing to the common denominator $H_\alpha^{2t-2}$ adds at most $(2t-2)(D+B)$, since $\deg H_\alpha\le D+B$.  For $t\ge2$, the inequality $1+4(t^2-2t+2)+(2t-2)\le4t^2$ closes the induction and gives
\[
 \deg N_t\le4(D+B)t^2.
\]
The case $t=1$ follows directly from \cref{w:eq:triangular-equation}.  This proves \cref{w:eq:rational-degree-bound}.

The same order-budget argument applies to every residual coefficient of \cref{w:eq:Q-Taylor-sub}.  There are no Taylor coefficients beyond $c_d$, at most $B$ noninitial coefficients occur in one monomial, and their total excess above $s$ is at most $Bd$.  Thus a common denominator uses only $O_B(d)$ powers of $H_\alpha$, while
\[
 \sum_j h_j^2\le\left(\sum_jh_j\right)^2\le B^2d^2.
\]
The product of the substituted numerators and the denominator clearing therefore have degree $O_B((D+B)(d+2)^2)$, which is \cref{w:eq:residual-degree} after increasing $C_B$.
\end{proof}

\subsubsection{Proof of the jet-decomposition theorem}

\begin{proof}[Proof of \cref{w:thm:jet-decomposition}]
We induct on the total $Y$-degree of $Q$.  If $d<r$, the ambient space $\A^{d+2}$ itself has dimension $d+2\le r+1$, so the theorem is immediate with one degree-one set and no exceptions.  We may therefore assume $d\ge r$.

If $Q$ is independent of every $Y_j$, then $Q(z,X)\equiv0$ has no solution for any $z$: $Z$-primitivity rules out a specialization at which all $X$-coefficients vanish.  This is the base case.

Otherwise, let $s\le r$ be maximal such that $Q$ depends on $Y_s$, and put
\[
 H:=\frac{\partial Q}{\partial Y_s}.
\]
Because $\deg_{Y_s}Q\le B<p$ whenever the theorem is used below, $H$ is nonzero.  (More generally, the assumption needed here is $\deg_{Y_s}Q<p$.)  Write
\[
 H=c(Z)\widetilde H
\]
where $c(Z)\in F[Z]$ is the content and $\widetilde H$ is $Z$-primitive.  Add the roots of $c$ to the exceptional set.  There are at most $D$ of them.

Every nonexceptional solution $(z,P)$ is of one of two types.

\smallskip
\noindent\emph{Singular type.}
If
\[
 H\bigl(z,X,P,P^{[1]},\ldots,P^{[s]}\bigr)\equiv0,
\]
then $(z,P)\in\Sol_d(\widetilde H)$.  The total $Y$-degree has decreased by one, while all $Z$- and $X$-degrees remain at most $D$.  The inductive hypothesis supplies the required cover for all singular solutions outside another set of at most $(D+d+2)^{C_0(r,B-1)}$ parameter values.

\smallskip
\noindent\emph{Nonsingular type.}
Suppose instead that the substituted polynomial $H(z,X,P,\ldots,P^{[s]})$ is nonzero.  Its $X$-degree is at most
\begin{equation}\label{w:eq:H-sub-degree}
 D+(B-1)d.
\end{equation}
Choose a fixed set
\[
 \Omega\subseteq F,
 \qquad \abs{\Omega}=D+(B-1)d+1.
\]
Then some $\alpha\in\Omega$ satisfies
\[
 H_\alpha(z,P(\alpha),P^{[1]}(\alpha),\ldots,P^{[s]}(\alpha))\ne0.
\]
Fix this $\alpha$.  In Taylor coordinates $c_j=P^{[j]}(\alpha)$, \cref{w:lem:triangular,w:lem:rational-jet-degree} express every $c_{s+1},\ldots,c_d$ as a rational function of the base variables
\[
 u=(z,c_0,\ldots,c_s)
\]
with numerator degree at most $E_0=O_B((D+B)(d+2)^2)$ and a common denominator $H_\alpha^L$ satisfying $\deg H_\alpha,L\deg H_\alpha\le E_0$.

The base variables must satisfy
\begin{equation}\label{w:eq:base-equation}
 Q(z,\alpha,c_0,\ldots,c_s)=0,
\end{equation}
and the numerators of all residual coefficient equations.  Let $U_\alpha$ be the corresponding constructible set with $H_\alpha\ne0$.  If $U_\alpha$ is nonempty, then $H_\alpha$ is not the zero polynomial.  Consequently the polynomial
\[
 Q_\alpha(Z,c_0,\ldots,c_s):=Q(Z,\alpha,c_0,\ldots,c_s)
\]
is nonzero, since $\partial Q_\alpha/\partial c_s=H_\alpha$.  As $\overline{U_\alpha}\subseteq Z(Q_\alpha)\subseteq\A^{s+2}$, we have
\begin{equation}\label{w:eq:base-dim}
 \dim\overline{U_\alpha}\le s+1\le r+1.
\end{equation}
Let $W_\alpha$ be the closed set defined by the base equation and the cleared residual equations. Then $U_\alpha=W_\alpha\cap\{H_\alpha\ne0\}$, so $\overline{U_\alpha}$ is the union of the irreducible components of $W_\alpha$ not contained in $Z(H_\alpha)$. By \cref{w:lem:effective-intersection} and \cref{w:eq:residual-degree},
\begin{equation}\label{w:eq:base-degree}
 \cdeg(\overline{U_\alpha})
 \le (D+d+2)^{O_{r,B}(1)}.
\end{equation}
The rational Taylor-coefficient map from $U_\alpha$ to $(z,c_0,\ldots,c_d)$ satisfies the hypotheses of \cref{w:lem:rational-graph}.  Its graph closure therefore has dimension at most $s+1$ and cumulative degree $(D+d+2)^{O_{r,B}(1)}$. The graph is identified with its $(z,c_0,\ldots,c_d)$ coordinates, since these already contain the base variables. The change from Taylor coefficients at $\alpha$ to monomial coefficients is invertible and linear, so it preserves dimension and degree.  We obtain one admissible algebraic set for each $\alpha\in\Omega$.

There are only $O_B(D+d)$ choices of $\alpha$, and the singular induction has depth at most $B$.  Increasing one effective exponent $C_0(r,B)$ absorbs all exceptional parameter values, all graph degrees, and all recursive contributions, proving \cref{w:eq:jet-complexity,w:eq:jet-dimension}.
\end{proof}

\begin{remark}[Why the large-characteristic hypothesis is essential]
If $d\ge p$, even the equation $Y_1=0$ has every polynomial in $X^p$ as a solution, so the solution locus has dimension growing with $d/p$.  The condition $d<p$ eliminates exactly this Frobenius freedom.  In the prime-field Reed--Solomon application, $d=k-1<n\le q=p$.
\end{remark}

\section{List bounds independent of the field size}\label{sec:list-bounds}
The root solver bounds the number of all polynomial solutions of $Q=0$ by a power of $q$. Agreement supplies additional linear equations, allowing a bound in terms of $n$ instead. This section records the consequence of combining interpolation with the geometric decomposition.

\begin{corollary}[A parameter-free solution cover]\label{cor:static-jet}
Fix $r,B\ge0$. There is an effective $C_0=C_0(r,B)$ such that the following holds. Let $F$ be algebraically closed of characteristic $p$, let $d<p$, let $B<p$, and let $D\ge2$. For a nonzero
\[
 Q(X,Y_0,\ldots,Y_r)\in F[X,Y_0,\ldots,Y_r],\qquad
 \deg_XQ\le D,\quad \deg_Y^{\rm tot}Q\le B,
\]
the coefficient vectors of all degree-at-most-$d$ solutions of
\[
 Q(X,P,P^{[1]},\ldots,P^{[r]})\equiv0
\]
are contained in a union of affine algebraic sets $W_1,\ldots,W_t\subseteq\A^{d+1}$ with
\[
 \dim W_j\le r,\qquad
 \sum_j\cdeg(W_j)\le(D+d+2)^{C_0}.
\]
\end{corollary}
\begin{proof}
Introduce an indeterminate $Z$ and regard $Q$ as independent of $Z$. It is $Z$-primitive. Apply \cref{w:thm:jet-decomposition}, and decompose its covering sets into irreducible components. Choose $z_0\in F$ outside the exceptional set and outside the finitely many constant $Z$-values of vertical components. Such a choice exists because $F$ is infinite.

Intersect every nonvertical component with $Z=z_0$. This is a proper hyperplane intersection, so its dimension is at most $r$ and its cumulative degree is at most the degree of the component. Vertical components have empty intersection with this slice. Identifying the slice with $\A^{d+1}$ yields the stated degree and dimension bounds. Every solution $P$ belongs to this cover because $(z_0,P)$ solves the parameter-independent equation and $z_0$ is nonexceptional. This is an existence argument for a cover; the choice of $z_0$ is not an additional step of the decoding algorithm.
\end{proof}

\begin{lemma}[Counting points isolated by their support equations]\label{lem:support-isolation}
Let $F$ be algebraically closed, and let $V_1,\ldots,V_t\subseteq\A^N$ have dimension at most $s$ and total cumulative degree at most $\Delta$. Let $h_1,\ldots,h_M$ have degree at most $E\ge1$, with $M\ge1$. Suppose $T\subseteq\bigcup_jV_j$ is finite and every $x\in T$ is assigned a set $S_x\subseteq[M]$ such that
\[
 \bigcap_{i\in S_x} Z(h_i)=\{x\}.
\]
Then
\[
 |T|\le\Delta\sum_{a=0}^{s}(ME)^a.
\]
\end{lemma}
\begin{proof}
Assign each point to an irreducible component of the cover that contains it. At a positive-dimensional component $W$ through $x$, not every assigned equation can vanish identically: otherwise $W$ would be contained in their singleton intersection. Choose an assigned equation that does not vanish identically and pass to an irreducible component of its intersection with $W$ containing $x$. Dimension drops at every step, so a zero-dimensional component is reached after at most $s$ cuts.

For a fixed ordered tuple of $a$ equation indices, B\'ezout bounds the total degree of the resulting proper-intersection components by $\Delta E^a$. A zero-dimensional component contributes at most its degree many points. Summing over the at most $M^a$ tuples of length $a$, and then over $0\le a\le s$, gives the bound. Repetitions or overlaps can only overcount points.
\end{proof}

\begin{theorem}[A sparse list bound from a differential interpolant]\label{thm:geometric-list}
Fix $r,B\ge0$, and let $\F_q$ be a finite field of characteristic $p>\max\{k-1,B\}$. Let $D=\{\alpha_1,\ldots,\alpha_n\}\subseteq\F_q$ have $n$ distinct points, let $S_i\subseteq\F_q$ have total incidence $M$, and let $A\ge k\ge1$. Suppose every $P\in\cL(D,k,\mathbf S,A)$ satisfies
\[
 Q(X,P,P^{[1]},\ldots,P^{[r]})\equiv0
\]
for a nonzero polynomial $Q$ with $\deg_XQ\le D_0$, $\deg_Y^{\rm tot}Q\le B$, and $D_0\ge2$. Then
\begin{equation}\label{eq:field-independent-list}
 |\cL(D,k,\mathbf S,A)|
 \le (D_0+k+1)^{C_0(r,B)}\sum_{a=0}^{r}M^a.
\end{equation}
If $M=0$, the list is empty.
\end{theorem}
\begin{proof}
Apply \cref{cor:static-jet} over $\overline{\F_q}$ with $d=k-1$. The candidate coefficient vectors lie in a cover of dimension at most $r$ and cumulative degree at most $(D_0+k+1)^{C_0(r,B)}$.

For every incidence $(i,b)$ with $b\in S_i$, define the hyperplane
\[
 h_{i,b}(\mathbf p)=\sum_{a=0}^{k-1}p_a\alpha_i^a-b.
\]
To a candidate $P$ assign all pairs $(i,P(\alpha_i))$ at its list-hit coordinates. There are at least $A\ge k$ distinct such coordinates. Any $k$ of their evaluation equations form an invertible Vandermonde system, so their common zero set, and hence that of all assigned equations, is exactly the coefficient vector of $P$. Apply \cref{lem:support-isolation} with $E=1$ and $s=r$. This proves \cref{eq:field-independent-list}.
\end{proof}

\begin{proof}[Proof of \cref{thm:capacity}]
The algorithmic assertion is \cref{cor:algorithmic-capacity}. For the list bound, take $A=k+\ceil{\gamma n}$ and choose the parameters of \cref{lem:uniform-padding}. By \cref{cor:interpolation-interface}, there is a nonzero differential interpolant containing every candidate, with
\[
 \deg_XQ<mA=O_{\gamma,L}(n),\qquad
 \deg_Y^{\rm tot}Q\le B_{\gamma,L}.
\]
For sufficiently large $n$, the characteristic $p=q\ge n$ exceeds both $k-1$ and $B_{\gamma,L}$. Use \cref{thm:geometric-list} with $D_0=\max\{2,mA\}$ and $M\le Ln$. Since $r,B_{\gamma,L},L$ are fixed, its right-hand side is $n^{O_{\gamma,L}(1)}$. Increasing an effective exponent and the length threshold absorbs the fixed multiplicative constants.

For singleton lists, relative distance at most $1-k/n-\gamma$ is equivalent to at least
\[
 n-\floor{(1-k/n-\gamma)n}=k+\ceil{\gamma n}
\]
agreements. This proves the decoding-radius formulation as well.
\end{proof}

\begin{remark}[What the list bound does not give]
The estimate is uniform in the field size; it is not a constant list-size bound. Nor does it by itself turn Kopparty's root enumeration into a $\operatorname{poly}(n,\log q)$ algorithm. The geometric cover is used for counting, whereas the stated decoder may enumerate $q^{O(r)}$ differential-equation roots before filtering.
\end{remark}

\section{Support-preserving descent and mutual correlated agreement}\label{sec:mca-proof}\subsection{Descent along the full agreement support}

The point-isolation argument has a parametric counterpart. We follow the complete support assigned to each witness and stop either on a vertical component or on a branch that explains that support identically. This is the step needed for MCA rather than ordinary correlated agreement.

\begin{lemma}[Support-preserving descent]\label{w:lem:support-descent}
Let $F$ be algebraically closed, let $\pi:\A^N\to\A^1$ be one coordinate function, and let $V_1,\ldots,V_t\subseteq\A^N$ be affine algebraic sets satisfying
\[
 \dim V_j\le s,
 \qquad
 \sum_{j=1}^{t}\cdeg(V_j)\le\Delta.
\]
Let $h_1,\ldots,h_n\in F[x_1,\ldots,x_N]$ have degree at most $L\ge1$.  Let $T\subseteq\bigcup_jV_j$ be finite, and assign to each $x\in T$ a set $S_x\subseteq[n]$ such that $h_i(x)=0$ for every $i\in S_x$.

Call $x$ \emph{branch-explained} if there is a positive-dimensional irreducible subvariety $W\subseteq V_j$ through $x$ for some $j$ such that $\pi|_W$ is nonconstant and
\[
 h_i|_W\equiv0
 \qquad\text{for every }i\in S_x.
\]
Then
\begin{equation}\label{w:eq:support-descent-bound}
 \abs{\{\pi(x):x\in T\text{ is not branch-explained}\}}
 \le \Delta\sum_{a=0}^{s}(nL)^a.
\end{equation}
\end{lemma}

\begin{proof}
Decompose all $V_j$ into irreducible components and assign each point of $T$ to one component containing it.  Starting from the assigned component $W_0$, process a point $x$ as follows.

If $W_a$ is zero-dimensional, stop.  If $\pi$ is constant on $W_a$, stop.  If every $h_i$ with $i\in S_x$ vanishes identically on $W_a$, then $x$ is branch-explained and the process stops.  Otherwise choose an index $i_{a+1}\in S_x$ for which $h_{i_{a+1}}$ is not identically zero on $W_a$, and let $W_{a+1}$ be an irreducible component of
\[
 W_a\cap Z(h_{i_{a+1}})
\]
containing $x$.  Because $h_{i_{a+1}}(x)=0$ but does not vanish identically on the irreducible variety $W_a$, every such component has strictly smaller dimension.  Thus an un-explained point reaches a zero-dimensional or $\pi$-vertical terminal component after at most $s$ cuts.

Fix an ordered index tuple $(i_1,\ldots,i_a)\in[n]^a$.  Organize the successive intersections with $Z(h_{i_1}),\ldots,Z(h_{i_a})$ as a tree of irreducible components, retaining only proper intersections along each branch.  B\'ezout's theorem bounds the sum of the degrees at depth $a$ by $\Delta L^a$.  Each zero-dimensional terminal component contributes at most its degree many values of $\pi$.  Each positive-dimensional vertical irreducible component contributes exactly one value of $\pi$, and the number of such components is at most the same cumulative-degree bound.  Hence all terminal components reached by this fixed tuple contribute at most $\Delta L^a$ parameter values.  There are $n^a$ ordered tuples of length $a$.  Summing over $0\le a\le s$ proves \cref{w:eq:support-descent-bound}.
\end{proof}

For $\mathbf p=(p_0,\ldots,p_{k-1})$, write
\[
 P_{\mathbf p}(X)=\sum_{a=0}^{k-1}p_aX^a.
\]
For the coefficient words in \cref{w:thm:mca}, define
\begin{equation}\label{w:eq:eval-hypersurface}
 h_i(z,\mathbf p)
 :=P_{\mathbf p}(\alpha_i)-\sum_{j=0}^{\ell}z^jf_{j,i}.
\end{equation}
These hypersurfaces have total degree at most $\ell$.

\begin{lemma}[A support-identical branch explains the full support]\label{w:lem:vandermonde-support}
Let $W\subseteq\A^{k+1}$, with coordinates $(z,\mathbf p)$, be a positive-dimensional irreducible variety over $\overline{\F_q}$.  Assume that $z|_W$ is nonconstant.  Let $S\subseteq[n]$ have $|S|\ge k$, and suppose
\[
 h_i|_W\equiv0
 \qquad\text{for every }i\in S.
\]
Then there are $P_0,\ldots,P_\ell\in\F_q[X]$, each of degree below $k$, such that
\[
 P_j(\alpha_i)=f_{j,i}
 \qquad(i\in S,\ 0\le j\le\ell).
\]
Moreover, at every point $(z_0,\mathbf p_0)\in W$, one has
\[
 P_{\mathbf p_0}=\sum_{j=0}^{\ell}z_0^jP_j.
\]
\end{lemma}

\begin{proof}
Choose $k$ indices $i_1,\ldots,i_k\in S$.  In the function field $\overline{\F_q}(W)$, their equations give
\[
 \begin{pmatrix}
  1&\alpha_{i_1}&\cdots&\alpha_{i_1}^{k-1}\\
  \vdots&\vdots&&\vdots\\
  1&\alpha_{i_k}&\cdots&\alpha_{i_k}^{k-1}
 \end{pmatrix}
 \mathbf p
 =
 \begin{pmatrix}
  \sum_{j=0}^{\ell}z^jf_{j,i_1}\\
  \vdots\\
  \sum_{j=0}^{\ell}z^jf_{j,i_k}
 \end{pmatrix}.
\]
The Vandermonde matrix is invertible over $\F_q$, so
\begin{equation}\label{w:eq:coefficient-curve}
 \mathbf p=\mathbf p^{(0)}+z\mathbf p^{(1)}+\cdots+z^\ell\mathbf p^{(\ell)}
\end{equation}
for vectors $\mathbf p^{(j)}\in\F_q^k$.  Let $P_j$ be the polynomial with coefficient vector $\mathbf p^{(j)}$.

Because the base field $\overline{\F_q}$ is algebraically closed and $z$ is nonconstant on the irreducible variety $W$, the element $z\in\overline{\F_q}(W)$ is transcendental over $\overline{\F_q}$.  For any $i\in S$, substituting \cref{w:eq:coefficient-curve} into the identity $h_i=0$ yields
\[
 \sum_{j=0}^{\ell}z^j\bigl(P_j(\alpha_i)-f_{j,i}\bigr)=0.
\]
Transcendence forces every coefficient to vanish, proving the first claim.  The identity in \cref{w:eq:coefficient-curve} holds in the regular coordinate functions on $W$, so it also gives the asserted polynomial identity at every point of $W$.
\end{proof}

\subsection{Proof of the MCA theorem}

\begin{proof}[Proof of \cref{w:thm:mca}]
Set
\[
 e=\floor{\rho n},
 \qquad A=n-e.
\]
The radius hypothesis gives
\begin{equation}\label{w:eq:A-ge-k}
 A\ge k+\ceil{\gamma n}\ge k.
\end{equation}
For each bad parameter $z\in\Bad_{\cC,\rho}^{(\ell)}(\mathbf f)$, choose one witnessing polynomial $P_z$ of degree below $k$, and put
\[
 S_z:=\Agr(y(z),P_z).
\]
Thus $|S_z|\ge A$, and $S_z$ is not a maximal $\mathbf f$-correlated set.  Let $x_z=(z,\mathbf p_z)\in\A^{k+1}$ be the coefficient point of $P_z$.  We select only one point above each bad parameter.

Construct the primitive symbolic differential interpolant $Q$ from \cref{w:prop:symbolic-Q}.  Every $x_z$ belongs to $\Sol_{k-1}(Q)$.  Write $B_\gamma=B_{\gamma,1}$ from \cref{lem:uniform-padding}, and enlarge $n_0$ so that $n>B_\gamma$ and $n\ge2$, and apply \cref{w:thm:jet-decomposition} over $\overline{\F_q}$ with $d=k-1$ and
\[
 D_0:=\max\{2,\deg_ZQ,\deg_XQ\}=O_{\gamma,\ell}(n).
\]
Indeed, $d=k-1<q$, while $B\le B_\gamma<n\le q$.  Moreover, $Z$-primitivity is preserved after extending scalars to $\overline{\F_q}$: a common linear factor over the algebraic closure would force its minimal polynomial over $\F_q$ to divide every coefficient of $Q$.  Thus all hypotheses of \cref{w:thm:jet-decomposition} hold.  We obtain an exceptional set $E_Q$ and varieties $V_1,\ldots,V_t\subseteq\A^{k+1}$ satisfying
\begin{equation}\label{w:eq:jet-cover-application}
 |E_Q|+\sum_{j=1}^{t}\cdeg(V_j)\le n^{C_1},
 \qquad
 \dim V_j\le s:=r+1,
\end{equation}
for an effective $C_1=C_1(\gamma,\ell)$, and every $x_z$ with $z\notin E_Q$ lies in $\bigcup_jV_j$.

Apply \cref{w:lem:support-descent} to the selected points with $z\notin E_Q$, using the support assignment $x_z\mapsto S_z$ and the hypersurfaces in \cref{w:eq:eval-hypersurface}.  Suppose that one selected point $x_z$ were branch-explained.  Then there would be a positive-dimensional nonvertical irreducible $W$ through $x_z$ on which every $h_i$, $i\in S_z$, vanishes identically.  By \cref{w:eq:A-ge-k,w:lem:vandermonde-support}, the entire exact support $S_z$ would be $\mathbf f$-correlated.  By \cref{w:lem:exact-maximal}, it would then be maximal, contradicting the choice of $P_z$ as an MCA witness.  Consequently every selected bad point outside $E_Q$ is not branch-explained.

Since the $h_i$ have degree at most $\ell$, \cref{w:lem:support-descent,w:eq:jet-cover-application} give
\[
 \abs{\Bad_{\cC,\rho}^{(\ell)}(\mathbf f)\setminus E_Q}
 \le n^{C_1}\sum_{a=0}^{r+1}(\ell n)^a.
\]
There are at most $|E_Q|\le n^{C_1}$ additional bad parameters in the exceptional set.  Because $r$ and $\ell$ are fixed, increasing an effective exponent gives
\[
 \abs{\Bad_{\cC,\rho}^{(\ell)}(\mathbf f)}\le n^{C(\gamma,\ell)}.
\]
The bound is uniform in the coefficient words.  Taking $\ell=1$, maximizing over $f_0,f_1$, and dividing by $q$ proves the MCA error estimate.
\end{proof}

\begin{corollary}[Ordinary correlated proximity gaps]\label{w:cor:ordinary-pg}
Under the hypotheses of \cref{w:thm:mca}, if more than $n^{C(\gamma,\ell)}$ values $z\in\F_q$ satisfy
\[
 \dist\!\left(\sum_{j=0}^{\ell}z^jf_j,\cC\right)\le\rho,
\]
then there are degree-below-$k$ polynomials $P_0,\ldots,P_\ell$ agreeing jointly with $f_0,\ldots,f_\ell$ on at least $A_\rho$ coordinates. Consequently, $\dist(y(z),\cC)\le\rho$ for every $z\in\F_q$.
\end{corollary}

\begin{proof}
Choose a nearby polynomial at every good parameter.  Since there are more good parameters than bad ones, at least one good parameter is not MCA-bad.  The exact support of its nearby polynomial belongs to $\cA_\rho(\mathbf f)$ and therefore is jointly explained by $P_0,\ldots,P_\ell$ on at least $A_\rho$ coordinates. For every parameter $z$, the polynomial $\sum_j z^jP_j$ agrees with $y(z)$ on that same support, proving the last assertion.
\end{proof}

\section{Consequences and quantitative scope}\label{sec:consequences}
\subsection{Coordinate transformations, punctures, and erasures}
\begin{corollary}[Coordinate-affine generalized RS codes]\label{v:cor:grs}
Let $v_i\in\F_q^\times$ and $b_i\in\F_q$.  All preceding algorithms apply to codewords
\[
  (v_iP(\alpha_i)+b_i)_{i=1}^n.
\]
For list recovery, replace $S_i$ by
\[
  S_i':=\{(z-b_i)/v_i:z\in S_i\}.
\]
This preserves both incidence size and agreement.
\end{corollary}

Puncturing is already allowed because the evaluation set is arbitrary.  Erased coordinates may simply be omitted; the theorem is then applied to the remaining evaluation set, with its new length $n'$, and with the desired absolute number of nonerased agreements.

When $b_i=0$, the same transformation preserves the exact agreement supports relevant to MCA: dividing each curve coefficient $f_{j,i}$ by $v_i$ carries joint explanations and maximality bijectively to those for ordinary Reed--Solomon codes. Arbitrary translations are allowed above only for list decoding and list recovery.

\subsection{Fixed interleaving}
For an integer $h\ge1$, let $\cC^{\equiv h}$ consist of $h$ rows from $\cC$, viewed as a length-$n$ code over $\F_q^h$. Distance counts columns on which the entire $h$-tuple differs. Write $\Lambda_h(\cC,\rho)$ for the maximum number of such codewords in a radius-$\rho$ ball.

\begin{corollary}[Polynomial lists for fixed interleaving]\label{cor:interleaving}
Under the singleton-list hypotheses of \cref{thm:capacity}, for every integer $h\ge1$,
\[
 \Lambda_h\!\left(\cC,1-k/n-\gamma\right)
 \le n^{hC_{\rm list}(\gamma,1)}.
\]
For fixed $\gamma,h$, the interleaved list can be enumerated deterministically in $q^{O_{\gamma,h}(1)}$ field operations.
\end{corollary}
\begin{proof}
Fix a received $h\times n$ array. Every valid interleaved codeword has the requisite agreement in every row separately. Its row polynomials therefore belong to the $h$ scalar lists from \cref{thm:capacity}, each of size at most $n^{C_{\rm list}(\gamma,1)}$. The product of these lists contains every valid interleaved codeword. Enumerate that product and test common-column agreement. The product bound proves the size estimate, and $q\ge n$ gives the stated running time for fixed $h$.
\end{proof}

\subsection{Fixed-slack certificates for the proximity challenges}
The preliminary Grand MCA and Grand List Decoding challenges consider smooth evaluation domains at rates
\[
 R\in\{1/2,1/4,1/8,1/16\},
\]
a target $\eps^*$ such as $2^{-128}$, and, for list decoding, a fixed interleaving order~\cite{ABF26,Prize26}. Smoothness is not needed by our theorems.

\begin{corollary}[Simultaneous asymptotic certificates]\label{cor:safe-certificates}
Fix $R\in(0,1)$, $\gamma\in(0,1-R)$, an integer $h\ge1$, and $\lambda>0$. Let $k=Rn$ be an integer and let $\cC=\RS_{\F_q}(D,k)$ for any $n$ distinct points in a prime field. For all sufficiently large $n$, define
\[
 C_*:=\max\{C(\gamma,1),\ hC_{\rm list}(\gamma,1)\},
\]
where $C(\gamma,1)$ is the MCA exponent in \cref{w:thm:mca}. If
\[
 q>2^\lambda n^{C_*},
\]
then, at $\rho=1-R-\gamma$,
\[
 \errMCA(\cC,\rho)<2^{-\lambda},\qquad
 \Lambda_h(\cC,\rho)<2^{-\lambda}q.
\]
\end{corollary}
\begin{proof}
The first estimate is \cref{w:thm:mca}. The second follows from \cref{cor:interleaving}; the common field-size assumption dominates both polynomial budgets.
\end{proof}

At the four designated rates and $\lambda=128$, the certified radii are
\begin{center}
\begin{tabular}{@{}cc@{}}
\toprule
Rate & Fixed-slack certified radius\\
\midrule
$1/2$ & $1/2-\gamma$\\
$1/4$ & $3/4-\gamma$\\
$1/8$ & $7/8-\gamma$\\
$1/16$ & $15/16-\gamma$\\
\bottomrule
\end{tabular}
\end{center}
These are lower-bound certificates for the safe radius, not a determination of the largest safe radius. In particular, the interleaved-list certificate uses the geometric $n^{O(1)}$ bound; it would not follow from the root solver's $q^{O(1)}$ bound alone.

\subsection{Field characteristic and the endpoint}
The interpolation construction is characteristic-free. The deterministic decoder invokes the prime-field version of Kopparty's root solver. The geometric list and MCA arguments extend to a finite field of characteristic
\[
 p>\max\{k-1,B_{\gamma,L}\}
\]
for the list bound, or with $L=1$ for MCA, provided the evaluation points are distinct. Interpolation over that field and its rational-function extension is already supplied by \cref{cor:interpolation-interface}. The two characteristic requirements have different purposes: $p>k-1$ prevents Frobenius freedom in Hasse--Taylor lifting, and $p>B_{\gamma,L}$ prevents the highest-variable partial derivative of the interpolant from vanishing identically.

None of the theorems permits replacing fixed $\gamma$ by an $n$-dependent slack without tracking the constants and the length threshold. The results in~\cite{CS25,KKH26} show that the literal up-to-capacity MCA formulation requires qualification; in particular, the prime-field smooth-domain constructions in~\cite{KKH26} occur in a shrinking near-capacity window. The fixed-slack statements here do not address that window or the exact integer-error boundary.

\subsection{Effectivity and remaining quantitative questions}
The parameters are effective: the order $r$ is chosen by \cref{w:eq:r-choice}, $B_{\gamma,L}$ is bounded in \cref{lem:uniform-padding}, and the geometric exponent follows from the degree recurrence in \cref{w:lem:rational-jet-degree} and constant-dimensional intersections. The support-isolation and support-descent arguments add only constantly many powers of the incidence size or block length.

The resulting constants are not optimized. No competitive numerical exponent or useful certificate for a prescribed small block length is claimed. Both the field-size inequality and the sufficiently-large-length threshold must be checked for a concrete application. The decoder is polynomial in $q$, not in $\log q$, and the list bound is polynomial in $n$, not constant. Improving the gap dependence, obtaining small-field MCA bounds, and converting the bounded-complexity geometry into a faster output-sensitive decoder remain separate quantitative questions.

\section*{AI Disclosure}

The interactions with ChatGPT 5.6 Pro were crucial in obtaining the main results of this paper. ChatGPT 5.6 and 6 Pro  were also used an editorial assistant. The authors take full responsibility for the results and have verified all claims, proofs, and references.

%\clearpage

\appendix

\section{Transfer to other root-enumeration settings}\label{v:app:transfer}

\begin{theorem}[Root-solver transfer]\label{v:thm:transfer}
Let $K>r\ge1$ and $1\le k\le K$.  Suppose a field $\F$ has an algorithm that enumerates all degree-below-$K$ solutions of
\[
  Q(X,P,P^{[1]},\ldots,P^{[r]})\equiv0
\]
for every nonzero $Q$ with individual $Y$-degree at most $B$ and
\[
  \deg_{(1,K-1,K-2,\ldots,K-r-1)}Q<D,
\]
in time $T_{\rm root}$ and with at most $L_{\rm root}$ outputs.  Under \cref{v:eq:alg1,v:eq:alg2}, sparse list recovery over $\F$ with actual degree bound $k-1$ is solvable using
\[
  r^{O_\vartheta(r)}\operatorname{poly}(n,M,B)+T_{\rm root}
 +L_{\rm root}\operatorname{poly}(n,K,M)
\]
field operations, with output list at most $L_{\rm root}$, provided $D\ge r^3A$ and the assumed root solver applies to the interpolant constructed in \cref{v:prop:optimized-interpolant}.
\end{theorem}

\begin{proof}
The interpolation construction, dimension estimate, local rank bound, and multiplicity argument are field-independent.  Construct a nonzero interpolant, invoke the assumed root solver, and filter its output by the actual degree bound and list agreement. The last summand explicitly accounts for verification of all returned roots; the transfer does not assume that this cost is already included in $T_{\rm root}$.
\end{proof}

\section{Quantitative derivative order and parameter ledger}\label{app:quantitative}
\subsection{Inverting the sparse-incidence condition}
The uniform fixed-slack statement is convenient for the main theorems. The following more detailed form retains the dependence on a target agreement $\eps$ and a multiplicative slack $\vartheta$.

\begin{corollary}[Bounded-list recovery at every constant agreement]\label{v:cor:bounded-list}
Fix $\vartheta\in(0,1)$, $\eps\in(0,1)$, and $L\ge1$.  Let
\[
  A=\ceil{\eps n},
  \qquad
  K=\ceil{(1-\vartheta)A},
\]
and suppose $k\le K$.  For all sufficiently large $n$, there is an effective choice
\begin{equation}\label{v:eq:r-choice}
 r=O_\vartheta\!\left(
  \left(1+\frac{L}{\eps}\right)^{1/\vartheta}
  \log^{5/\vartheta}\!\left(e+\frac{L}{\eps}\right)
 \right)
\end{equation}
for which every instance with $\abs{S_i}\le L$ is list recoverable up to $A$ hits by the algorithm of \cref{v:thm:prime-sparse}.  Over every prime $q\ge n$ satisfying the degree conditions, the time is $q^{O_\vartheta(r)}$ and the list size is at most $q^{4r+6}$.
\end{corollary}

\begin{proof}
Because $K-1<(1-\vartheta)A$,
\[
  \frac{A}{K-1}>\frac{1}{1-\vartheta}.
\]
Also $M\le L n$ and $K-1=\Omega_\vartheta(\eps n)$.  Thus
\[
  \frac{M}{K-1}=O_\vartheta\!\left(\frac{L}{\eps}\right).
\]
Let $x=1+L/\eps$.  For a sufficiently large hidden constant in \cref{v:eq:r-choice},
\[
  \frac{r^\vartheta}{\log^4(er)}\ge C_\vartheta x,
\]
so \cref{v:eq:alg2} holds.  Since $r$ is fixed once $\vartheta,\eps,L$ are fixed, the remaining inequalities $r<K$, $B<q$, and $r^3A<q^2$ hold for all sufficiently large $n$ whenever $q\ge n$.
\end{proof}

In particular, at rate at most $(1-\vartheta)\eps$, singleton lists need only
\[
 r=O_\vartheta\!\left((1+1/\eps)^{1/\vartheta}
           \log^{5/\vartheta}(e+1/\eps)\right).
\]
The interpolation and root-enumeration stages both take $q^{O_\vartheta(r)}$ operations in their stated prime-field parameter range. This is a fixed-parameter polynomial-time bound; it is not polynomial dependence on $1/\vartheta$ or on the additive capacity gap.

\subsection{Rounding and field-size thresholds}
At the actual capacity-minus-$\gamma$ radius, set $A=k+\ceil{\gamma n}$ and $K=\ceil{(1-\gamma/2)A}$. The upper bound $k\le(1-\gamma)n$ implies $A\le n$, while \cref{lem:uniform-padding} gives $K\ge k$ and $K-1<(1-\gamma/2)A$. The same lemma applies at smaller radii by increasing $A$ up to $n$.

After $\gamma,L$ are fixed, $r,m,B_{\gamma,L}$ are constants and $K=\Theta_\gamma(n)$. Thus $r<K$, $B<q$, and $r^3A<q^2$ hold for sufficiently large $n$ whenever $q\ge n$. The stricter field-size budget in \cref{cor:safe-certificates} serves a different purpose: it makes a polynomial number of exceptional parameters, or a polynomial list size, small relative to $q$.

\subsection{Parameter ledger}
\begin{center}
\small
\begin{tabular}{@{}>{\raggedright\arraybackslash}p{0.13\textwidth}>{\raggedright\arraybackslash}p{0.28\textwidth}>{\raggedright\arraybackslash}p{0.49\textwidth}@{}}
\toprule
Parameter & Choice & Role\\
\midrule
$\gamma$ & Fixed additive slack & Gives $A-k\ge\gamma n$ and a rate-uniform interpolation budget.\\
$L$ & Input-list bound & Implies $M\le Ln$; constants for recovery depend on $L$.\\
$\ell$ & Curve degree & Bounds the degree of evaluation hypersurfaces in MCA.\\
$A$ & Required number of hits & Supplies root multiplicity at least $mA$.\\
$K$ & $\ceil{(1-\vartheta)A}$ & Padded degree ceiling, at least the actual dimension $k$.\\
$\vartheta$ & $\gamma/2$ for uniform padding & Fixed multiplicative degree margin.\\
$r,m$ & $r$ from \cref{w:eq:r-choice}; $m=r^3$ & Differential order and local multiplicity.\\
$c,s$ & $(1-\vartheta)^{-1}-2/\log(er)$ and $(1-\vartheta)^{-1}-1/\log(er)$ & Weighted and ordinary tail cutoffs.\\
$W$ & $\floor{crm/\log(er)}$ & Weighted derivative-tail budget.\\
$B$ & $\ceil{mA/(K-1)}$ & Total derivative-variable degree bound.\\
$D_0$ & $O(n)$ at fixed parameters & Coefficient-degree scale in the jet theorem.\\
$E_Q$ & Finite exceptional set & Values lost under derivative-content specialization.\\
$\Delta$ & Polynomial in $n$ & Total cumulative degree of the solution cover.\\
$S_P,S_z$ & Full list-hit or exact agreement support & Determines all cuts; no agreement subset is substituted for an MCA witness's support.\\
\bottomrule
\end{tabular}
\end{center}

\end{document}